\documentclass[a4paper,11pt]{article}

\usepackage[english]{babel}
\usepackage[utf8]{inputenc}
\usepackage[T1]{fontenc}

\usepackage{graphicx} 
\usepackage[margin=1in]{geometry}
\usepackage{algorithm}
\usepackage[noend]{algpseudocode}
\usepackage{amsthm}
\usepackage{amssymb,amsmath}  
\usepackage{thmtools}  
\usepackage{mathtools}
\usepackage[dvipsnames]{xcolor}
\definecolor{VUBlue}{rgb}{0,0.466,0.702}
\usepackage[colorlinks=true, allcolors=VUBlue]{hyperref}
\usepackage[capitalize]{cleveref}
\usepackage{todonotes}
\usepackage{placeins}
\usepackage{subcaption}
\usepackage{booktabs}
\usepackage{microtype}
\usepackage[normalem]{ulem}
\usepackage{newtxtext}
\usepackage{newtxmath}

\usepackage{enumitem}
\setlist[1]{labelindent=\parindent}
\setlist[enumerate]{label=(\arabic*)}
\setlist[itemize]{noitemsep}

\usepackage{authblk}

\author[1]{Sebastian Angrick}
\author[2,3]{Ben Bals}
\author[4]{Paweł Gawrychowski}
\author[5]{Solon P. Pissis}
\author[6]{Yuki Yonemoto}
\affil[1]{Karlsruhe Institute of Technology, Germany}
\affil[2]{CWI, Amsterdam, The Netherlands}
\affil[3]{Vrije Universiteit, Amsterdam, The Netherlands}
\affil[4]{University of Wrocław, Poland}
\affil[5]{The Cyprus Institute, Nicosia, Cyprus}
\affil[6]{Kyushu University, Japan}
\date{\vspace{-.5cm}}

\declaretheorem[numberwithin=section]{theorem}

\declaretheorem[sibling=theorem]{corollary}
\declaretheorem[sibling=theorem]{lemma}
\declaretheorem[sibling=theorem]{definition}

\declaretheorem[sibling=theorem]{remark}

\declaretheorem[sibling=theorem]{observation}

\declaretheorem[sibling=theorem,title=Open Question]{openquestion}

\declaretheorem[numbered=no,title=Intuitive Statement of \Cref{thm:td-lb}]{informalTheoremBMM}
\declaretheorem[numbered=no,title=Strong Exponential Time Hypothesis (SETH)~\cite{ImpagliazzoP01}]{seth}
\declaretheorem[numbered=no,title=Orthogonal Vectors Hypothesis (OVH)]{ovh}
\declaretheorem[numbered=no,title=Triangle Detection Hypothesis]{tdh}

\title{String Matching in (Block) Graphs: \\ A Full Classification by Walk Length} 

\usepackage{amsthm,amssymb,amsmath}  
\usepackage{mathtools}
\usepackage{graphicx}
\usepackage{xspace}
\usepackage{thmtools}
\usepackage{url}
\usepackage{comment}
\usepackage{subcaption}
\usepackage{bm}
\usepackage{multirow}
\usepackage{listings}

\usepackage{enumitem}
\setlist[1]{labelindent=\parindent}
\setlist[enumerate]{label=(\arabic*),noitemsep}

\usepackage{algorithm}
\usepackage[noend]{algpseudocode}

\def\dd{\mathinner{.\,.}}
\newcommand{\cO}{\mathcal{O}}
\newcommand{\cA}{\mathcal{A}}
\newcommand{\Oish}{\widetilde{\cO}}
\newcommand{\Omegaish}{\widetilde{\Omega}}
\newcommand{\Occ}{\mathrm{Occ}}

\newcommand{\per}{\textsf{per}}
\newcommand{\border}{\textsf{border}}
\newcommand{\setsize}[1]{\left|#1\right|}
\DeclareMathOperator{\degreeOp}{\delta}
\DeclareMathOperator{\indegreeOp}{\degreeOp^{-}}
\DeclareMathOperator{\outdegreeOp}{\degreeOp^{+}}
\newcommand{\indegree}[1]{\indegreeOp\hspace{-0.2em}\left(#1\right)}
\newcommand{\outdegree}[1]{\outdegreeOp\hspace{-0.2em}\left(#1\right)}
\newcommand{\N}{\mathbb{N}}

\newcommand{\matchVec}[1]{\vec{M_{#1}}}
\newcommand{\interMat}[1]{T_{#1}}
\newcommand{\matchMat}[1]{A_{V_{#1}}}
\newcommand{\resultMat}[1]{A'_{V_{#1}}}
\newcommand{\adjMat}[2]{\mathrm{Adj}_{V_{#1\to{}#2}}}

\newcommand{\smol}{\varepsilon}
\newcommand{\polylog}{\mathrm{polylog}}

\newcommand{\defproblem}[3]{
\vspace{2mm}
\noindent\fbox{
   \begin{minipage}{0.96\columnwidth}
   \textsc{#1}\\
   {\bf{Input:}} #2  \\
   {\bf{Output:}} #3
   \end{minipage}
   }
   \vspace{2mm}
}

\newcommand{\SMLG}{SMLG\xspace}
\newcommand{\hSMLG}{$h$-SMLG\xspace}
\newcommand{\SMBG}{SMBG\xspace}
\newcommand{\bSMBG}{$b$-SMBG\xspace}

\newcommand{\ie}[1]{(i.e.,~#1)}
\newcommand{\eg}[1]{(e.g.,~#1)}

\begin{document}

\maketitle

\begin{abstract}
We consider directed graphs in which the nodes are labeled with strings.
A walk in such a graph naturally corresponds to the concatenation of the
visited nodes' labels. These graphs are widely used in bioinformatics to compactly describe large collections
of highly similar genomes. Given such a graph $G=(V,E)$ and a pattern of length $m$,
we seek a walk whose corresponding string has an occurrence of the pattern.
We call this the \SMLG problem. Amir et al.~[J. Algorithms, 2000] showed that 
\SMLG can be solved in $\cO(m\setsize{E} + N)$ time, where $N$ is the total
length of all node labels. Equi et al.~[ACM Trans.~Algorithms, 2023] showed that this is essentially optimal (under SETH).

The existing lower bound assumes that the sought walk is of length $\Theta(\setsize{V})$.
Thus, we might be able to bypass this lower bound by restricting the walk length to $b-1$, which naturally reduces to having as input a directed graph whose set of nodes is partitioned into $b$ \emph{blocks}. Then, we seek a walk in this graph that starts in the first block and ends in the last block. We call this the \bSMBG problem.
Equi et al.~[Algorithmica, 2023] showed that, if we impose no restriction on $b$, the existing algorithm of Amir et al.~is essentially optimal for \bSMBG (again under SETH).
We provide a more fine-grained classification that essentially settles
the complexity of \bSMBG parameterized by $b$:
\begin{enumerate}
    \item For $b=2$, Pissis [SOSA 2025] already provided a simple $\cO(m + \setsize{E}+N)$-time algorithm.
    \item We design a new $\Oish(m + \setsize{E} + N)$-time algorithm for $b=3$. As a direct implication of this result, the \SMLG problem for $b\leq 3$ (walks of length at most $2$) also admits near-linear-time complexity.
    \item There is no $\cO((m\setsize{E})^{1-\smol} + N)$-time \emph{combinatorial} algorithm, for any $b \geq 4$ and $\smol>0$.
    \item There is an algorithm working in $\cO\left(\max(\setsize{V}, m)^\omega+N\right)$ time, where $\omega$ is the matrix multiplication exponent, which is conditionally optimal for graphs with $b \geq 4$ blocks.
    \item  Under SETH, no
    $\cO((m\setsize{E})^{1-\smol} + N)$-time algorithm exists, for any $b=\omega(\log\setsize{V})$ and $\smol > 0$.
\end{enumerate}

Although our motivation is primarily of a theoretical nature, we stress that our algorithms are simple to implement. As such, they may contribute to practical advancements in applications where the \SMLG problem is an important primitive, such as in the analysis of pangenome graphs.
\end{abstract}

\clearpage

\section{Introduction}
\label{sec:problem-def}

String matching is the classical problem of finding the occurrences of a \emph{pattern} $P=P[1\dd m]$ of length $m$ in a \emph{text} $T=T[1\dd n]$ of length $n$.
The problem is by now well-understood: there exist $\cO(n)$-time algorithms to solve the problem (e.g.,~\cite{DBLP:journals/siamcomp/KnuthMP77,DBLP:books/daglib/0020103}), even when using only $\cO(1)$ extra space (e.g.,~\cite{DBLP:journals/jacm/CrochemoreP91,DBLP:journals/tcs/BreslauerGM13}). Since much of today's data is interconnected, it is natural to study string matching not only in linear texts but also in \emph{labeled graphs}.
Indeed, large-scale labeled graphs are now prevalent across diverse areas, including graph databases~\cite{DBLP:journals/csur/AnglesG08}, graph mining~\cite{DBLP:journals/csur/ChakrabartiF06}, and, more recently, bioinformatics~\cite{DBLP:journals/nc/BaaijensBDVP22}. While most applications
require sophisticated operations on these graphs, they often rely on primitives that locate graph walks whose node labels match a given pattern~\cite{DBLP:journals/bioinformatics/RautiainenMM19,Rautiainen2020,DBLP:journals/almob/AsconeBCEGGP26}.\footnote{Indeed, string matching in node-labeled graphs is a central topic of large bioinformatics consortia such as ALPACA (\url{https://alpaca-itn.eu/}) or PANGAIA (\url{https://www.pangenome.eu/}).} See \Cref{fig:VG} for an example.

\begin{figure}[ht]
    \centering
    \includegraphics[width=0.8\linewidth]{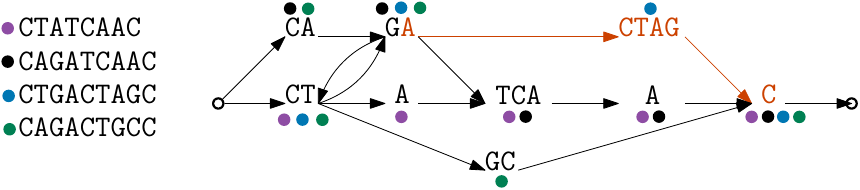}
    \caption{The DNA sequences on the left are represented by the node-labeled graph on the right. The pattern $P=P[1\dd 6]=\texttt{ACTAGC}$ is matched in a walk of length $2$ (in orange). This type of graph is known as \emph{variation graph}~\cite{DBLP:journals/nc/BaaijensBDVP22}, and it is a widely used representation for analyzing pangenomes.}
    \label{fig:VG}
\end{figure}
\subsection{String Matching in Graphs and the State of the Art}

Let $\Sigma=\{1,2,\ldots,\sigma\}=[1,\sigma]$ be an \emph{alphabet} of $\sigma$ elements that we call \emph{characters}. A \emph{node-labeled graph} $G = (V , E, \lambda)$ is a directed graph $(V , E)$ equipped with a \emph{labeling function} $\lambda : V \rightarrow \Sigma^*$ that defines which string over $\Sigma$ is assigned to each node as a \emph{label}. 
We set $N_G\coloneqq\sum_{v\in V}|\lambda(v)|$, the total length of the node labels in $G$; we may drop the subscript $G$ from $N_G$ when the context is clear. 
Given a pattern $P=P[1\dd m]$ of length $m$ over $\Sigma$, we say that $P$ has a \emph{match} in $G$ if there is a walk $v_1,\ldots, v_h$
in $G$ such that $P = \alpha\cdot \lambda(v_2) \cdots \lambda(v_{h-1}) \cdot\beta$, $\alpha$ is a suffix of $\lambda(v_1)$ and $\beta$ is a prefix of $\lambda(v_h)$; we also say that $P$ \emph{occurs} in $G$, and that $v_1,\ldots, v_h$ is an \emph{occurrence} of $P$ in $G$.

We next define the problem of string matching in node-labeled graphs (cf.~\cite{DBLP:journals/jal/AmirLL00,DBLP:journals/talg/EquiMTG23}).

\defproblem{String Matching in Labeled Graphs (\SMLG)}
{A node-labeled graph $G = (V , E, \lambda)$ and a pattern $P=P[1\dd m]$.}
{\emph{True} if and only if there is at least one occurrence of $P$ in $G$.}

The following upper bound is known due to Amir, Lewenstein, and Lewenstein~\cite{DBLP:journals/jal/AmirLL00}.

\begin{theorem}[\cite{DBLP:journals/jal/AmirLL00}]\label{the:SMLG-UB}
   The \SMLG problem can be solved in $\cO(m\setsize{E} + N)$ time.
\end{theorem}

The following lower bound, conditioned on the
Strong Exponential Time Hypothesis (SETH)~\cite{ImpagliazzoP01}, is known due to Equi, Mäkinen, Tomescu, and Grossi~\cite{DBLP:journals/talg/EquiMTG23}.

\begin{theorem}[\cite{DBLP:journals/talg/EquiMTG23}]\label{the:SMLG-LB}
   Under SETH,
   the \SMLG problem 
   cannot be solved in $\cO((m\setsize{E})^{1-\smol} + N)$ time, for any constant $\smol > 0$, even if every node label is a single character of a binary alphabet 
   and $G$ is a DAG, where the sum of the out- and in-degree of any node is at most 3. 
\end{theorem}

\subsection{Our Motivation and Parameterization}

Given \Cref{the:SMLG-UB} and \Cref{the:SMLG-LB}, we were thus motivated to ask whether \SMLG instances restricted to walks of \emph{bounded length} can be solved significantly faster than $\cO(m\setsize{E} + N)$ time. More concretely, what if we are looking for occurrences of $P$ in $G$ of the form $v_1,\ldots,v_{h}$ for a fixed $h\in \N$? Namely, we parameterize the \SMLG problem by walk length $(h-1)$. Pissis showed that this problem can be solved in $\cO(m+\setsize{E}+N)$ time for $h=2$~\cite{DBLP:conf/sosa/Pissis25} (for an earlier, slower solution for $h=2$, see~\cite[Lemma 8]{DBLP:journals/jda/Thachuk13}). Let us now formalize the problem.

\defproblem{$h$-String Matching in Labeled Graphs (\hSMLG)}
{A node-labeled graph $G = (V , E, \lambda)$ and a pattern $P=P[1\dd m]$.}
{\emph{True} if and only if there is at least one occurrence of $P$ in $G$ of the form $v_1,\ldots,v_h$.}

Furthermore, real--world instances of this problem \eg{arising from pangenome graphs~\cite{DBLP:journals/nc/BaaijensBDVP22}} often \underline{include large node labels}. For instance, a recent survey~\cite{Venkataramana} shows that the average node label length
is in the order of $10^3$ and for some graphs it is close to $10^6$.
Intuitively, this is because these graphs are constructed over a collection of \emph{highly similar} genomes. Thus, when the pattern $P$ is relatively short, it is reasonable to expect that an occurrence of $P$ may span only a few nodes in $G$ \ie{that the sought walk is somewhat short}. Note that patterns may be short in practice for two reasons: (i) they are either short DNA sequencing reads; or (ii) they are short DNA fragments used as anchors for alignments.

For algorithmic and application domain reasons, it is common to consider the more structured \emph{block graphs}. 
Let us give a formal definition from~\cite{DBLP:conf/wabi/MakinenCENT20}; see \Cref{fig:fgpm} for an example. 
The definitions and notation introduced for \SMLG carry over to block graphs. 

\begin{definition} \label{def:block-graph}
A \emph{block graph} $G=(V,E,\lambda)$
is a node-labeled graph satisfying:
\begin{enumerate}
    \item The set $V$ can be partitioned into a sequence $V_1,V_2,\ldots,V_{b}$ of \emph{blocks};
    \item If $(u,v)\in E$ then $u\in V_i$ and $v\in V_{i+1}$, for some $i\in[1,b)$.
\end{enumerate}   
When in addition to properties (1) and (2),
we have the property that, if for all $u,v\in V_i$, $i\in[b]$, we have $\lambda(u)\neq \lambda(v)$ for $u\neq v$, we call $G$ a \emph{block graph with unique labels}.
\end{definition}

\begin{figure}[ht]
    \centering
    \includegraphics[width=.35\linewidth]{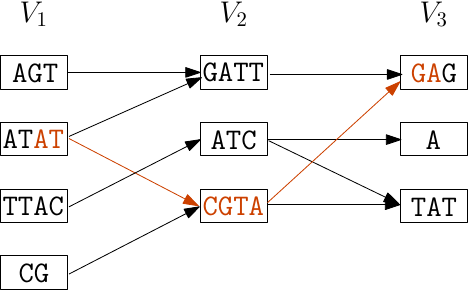}
    \caption{A block graph with an occurrence of $P=P[1\dd 8]=\texttt{ATCGTAGA}$ (in orange).}
    \label{fig:fgpm}
\end{figure}

\defproblem{String Matching in Block Graphs (\SMBG)}
{A block graph $G = (V , E, \lambda)$ with $b$ blocks and a pattern $P=P[1\dd m]$.}
{\emph{True} if and only if there is at least one occurrence of $P$ in $G$.}

When $h$ is bounded, we can efficiently reduce \hSMLG on $G=(V,E,\lambda)$ to string matching in a block graph
$G'=(V',E',\lambda')$ with $b\coloneqq h$ blocks of nodes; that is, $V'=V'_1\sqcup\ldots \sqcup V'_b$. Informally, each block $V'_i$ is set to $V$ ($G'$ gets $h$ copies of $V$), there is an edge in $E'$ from $v \in V'_i$ to $v' \in V'_{i+1}$ if and only if $(v,v')\in E$ ($G'$ gets $h-1$ copies of $E$), and the $\lambda'$ values are inherited from $\lambda$. Then, an occurrence $v_1,\ldots,v_h$ of $P$ exists in $G$ if and only if an occurrence $u'_1,\ldots,u'_b$ of $P$ exists in $G'$. 

\SMBG can be solved in $\cO(m\setsize{E} + N)$ time using \Cref{the:SMLG-UB}. Ascone et al.~\cite{DBLP:journals/almob/AsconeBCEGGP26} showed how to solve \SMBG
in $\cO(bm + N + \setsize{E})$ time for a specific class of block graphs: when all labels in each $V_i$ have the same length $k_i$. Another specific class of block graphs has also been investigated:
when, for all pairs $(V_i,V_{i+1})$, all possible edges are present. In this case,
the set $E$ becomes implicit, and the resulting graph is known as an \emph{elastic-degenerate string}~\cite{DBLP:journals/iandc/IliopoulosKP21}.
Many algorithms have been proposed for \SMBG on this class of block graphs~\cite{DBLP:journals/siamcomp/BernardiniGPPR22,DBLP:conf/cpm/AoyamaNIIBT18,DBLP:conf/cpm/GrossiILPPRRVV17}.
Similar to \Cref{the:SMLG-LB}, Equi et al.~\cite{DBLP:journals/algorithmica/EquiNACTM23} showed that \SMBG cannot be solved in $\cO((m\setsize{E})^{1-\smol}+N)$ time, for any $\smol>0$, unless SETH fails. 

As we are interested in the parameterized version $h$-\SMLG of the more general \SMLG problem, we define the corresponding version of string matching in block graphs.

\defproblem{$b$-String Matching in Block Graphs (\bSMBG)}
{A block graph $G = (V , E, \lambda)$ with $b$ blocks and a pattern $P=P[1\dd m]$.}
{\emph{True} if and only if there is at least one occurrence of $P$ in $G$ of the form $v_1,\ldots,v_b$.}

Since we seek an occurrence of $P$ in $G$ of the form $v_1,\ldots,v_b$, it must form a walk in $G$ of length exactly $(b-1)$, which means that the occurrence $v_1,\ldots,v_b$ spans all $b$ blocks of $G$.
Our algorithms and lower bounds can be adapted to lift this restriction to allow occurrences using some of the blocks. To simplify the presentation, we stick to this version of the problem.

Recall that \hSMLG for $h=2$ can be solved in linear time~\cite{DBLP:conf/sosa/Pissis25}.
This algorithm solves \bSMBG for $b=2$ in linear time and then applies the above reduction.
The contrast between this result and the general $(m\setsize{E})^{1-\smol}$ lower bound brings us to the following basic question:

\begin{center}
    \emph{Can we solve \bSMBG in $\cO((m\setsize{E})^{1-\smol}+N)$ time for some fixed $b\ge 3$ and some $\smol>0$?}
\end{center}

\subsection{Our Contributions and Techniques}

If we consider the set of all possible instances of the \bSMBG problem,
the algorithm by Amir et al.~\cite{DBLP:journals/jal/AmirLL00} is essentially optimal under SETH (up to subpolynomial improvements)~\cite{DBLP:journals/algorithmica/EquiNACTM23}. 
We thus take a finer look at the complexity landscape of \bSMBG.
Notably, the instances underlying this hardness~\cite{DBLP:journals/algorithmica/EquiNACTM23} are (1) \textbf{wide} (the number of blocks is in $\Theta(\setsize{V})$) and (2) \textbf{sparse} (they have a constant number of edges per node).
Thus, it is natural to ask whether we can improve on the state of the art for (1) \textbf{narrower}, or (2) \textbf{denser} instances.
Towards these directions, a simple linear-time algorithm is known for $b=2$ blocks~\cite{DBLP:conf/sosa/Pissis25}. We extend this regime by showing a near-linear-time algorithm for $b=3$ blocks.

\begin{restatable*}{theorem}{thmThreeBlocks}\label{thm:three-blocks}
The \bSMBG problem for $b=3$ can be solved in $\Oish(m + \setsize{E} + N)$ time. This directly implies that the \hSMLG problem for $h=3$ can be solved in $\Oish(m + \setsize{E} + N)$ time.    
\end{restatable*}

We emphasize that obtaining a near-linear-time algorithm for $b=3$ was \emph{not} a straightforward extension of the $b=2$ case. First, note that, for $b=1$, we can simply employ any linear-time string matching algorithm \eg{\cite{DBLP:journals/siamcomp/KnuthMP77}}. Then, for $b=2$, the primary challenge presented in~\cite{DBLP:conf/sosa/Pissis25} was not to obtain a near-linear-time algorithm, but rather to obtain a \emph{linear-time} algorithm. In particular, we can first compute, for every node $v$ in the first block, the longest prefix of $P$ matching a suffix of $\lambda(v)$, and then compute the symmetric information for the second block. Finally, we can simulate the prefix-suffix ``concatenation'' per edge using existing data structures~\cite{DBLP:conf/soda/GuFB94,DBLP:journals/talg/Gawrychowski13,DBLP:conf/sosa/Pissis25}. Unfortunately, when we move from $b=2$ to $b=3$, the existence of the middle block renders this technique inapplicable. 

To overcome this challenge, we rely on a well-known combinatorial fact about the \emph{borders} of a string \ie{about prefixes that are also suffixes of a string}: even though there can be a linear number of borders in a string \eg{consider the string \texttt{aa}\dots\texttt{a}}, they can be compactly represented as a logarithmic number of \emph{arithmetic progressions} (APs)~\cite{DBLP:journals/siamcomp/KnuthMP77,DBLP:journals/jct/GuibasO81,DBLP:conf/soda/GuFB94,DBLP:journals/talg/BreslauerG14,DBLP:journals/siamcomp/KociumakaRRW24,DBLP:conf/cpm/Boneh025}.
These APs are sets of indices of the form $\ell, \ell + x, \ell + 2x, \dots, \ell + kx = r$; 
each set has a periodic structure and can thus be represented by a tuple $(\ell, r, x)$ of the first element $\ell$, the last element $r$, and the common difference $x$.

For each node in the first and last blocks, we compactly represent the prefixes (or suffixes) of the pattern that match a node label as APs.
For middle-block nodes, we compute candidate matches by intersecting the APs of their $k$ in-neighbors with all occurrences of the middle node's label, doing so \emph{implicitly}---without computing this set.
Next, using a carefully designed intersection operation, we intersect these candidates with the suffix APs from the $k'$ out-neighbors.
To achieve this efficiently, we exploit periodicity information encoded in the APs. Crucially, these steps are no longer \emph{linear} in $m$: they output only a logarithmic (in $m$) number of APs per edge. By combining AP-based compression, efficient intersection, and constant-time verification, the algorithm achieves near-linear time.

As a second main result, we show a lower bound by reducing from Boolean matrix multiplication (BMM) to \bSMBG with $b \geq 4$ blocks.
As a consequence, a near-linear-time algorithm for this setting exists only if one exists for BMM. This pinpoints the threshold number of blocks for which near-linear-time algorithms are likely. Additionally, under the popular BMM conjecture~\cite{DBLP:conf/soda/AbboudWY15}, there can be no combinatorial algorithm polynomially faster than the state-of-the-art \SMBG algorithm.

\begin{informalTheoremBMM}\label{thm:lb-bmm}
    Unless the combinatorial BMM conjecture is false,
    we cannot solve \bSMBG in $\cO((m\setsize{E})^{1-\smol} + N)$ time, for any $b \geq 4$ and constant $\smol>0$ using a combinatorial algorithm.
\end{informalTheoremBMM}

For block graphs with unique labels, our lower bound holds when there are at least six blocks, thus leaving the corresponding problem open for the cases $b=4$ and $b=5$.

Towards (2), our goal for developing more efficient algorithms for dense graphs, we give a matrix-multiplication-based algorithm that can be much faster in this setting. 

\begin{restatable*}{corollary}{thmMMEasy}\label{thm:mm-easy}
  The \bSMBG problem can be solved in $\cO\left(\max(\setsize{V}, m)^{\omega}+N\right)$ time, where $\omega$ is the matrix multiplication exponent. 
\end{restatable*}

As our previous lower bound reduces to BMM, this algorithm is optimal up to subpolynomial factors.
Thus, our upper and lower bounds together essentially settle the complexity landscape of \bSMBG parameterized by the block number $b$. For more blocks, this algorithm can be analyzed more subtly (see \Cref{thm:mm-generality}). For example, if there are $\sqrt{\setsize{V}}$ blocks, each having equally many nodes, and $m \in \cO(\setsize{V})$, then the problem can be solved in $\cO(\setsize{V}^{2.126})$ time, using the current best bound $\omega \leq 2.371552$ on the matrix multiplication exponent $\omega$~\cite{DBLP:conf/soda/WilliamsXXZ24}, thus in less than $\setsize{V}^\omega \approx \setsize{V}^{2.372}$ time.

Interestingly, by relying on rectangular and sparse matrix multiplication, the algorithm also efficiently handles wider and sparser instances and in the worst case matches the state-of-the-art $\cO(m\setsize{E})$ running time.
For example, for the lower bound instances by Equi et al.~\cite{DBLP:journals/algorithmica/EquiNACTM23}, since they are so wide and sparse, there will be only a constant number of edges between any two consecutive blocks, thus leaving no room for algebraic improvement of this part of the algorithm.
Also note that in these instances, $\setsize{E} \in \cO(\setsize{V})$, meaning they provide little insight on the optimal dependence on $\setsize{E}$ versus $\setsize{V}$.
In some sense, this suggests the state of the art is the \emph{correct} algorithm, at least as long as we only consider the parameters $m$, $\setsize{E}$ and ignore $\setsize{V}$, $b$, which is the setting our work improves on.

We provide upper and lower bounds expressed in terms of $\setsize{V}$ or $\setsize{E}$.
Because the interplay between these parameters depends on graph density, direct comparisons between different bounds---and with the state of the art---can be nuanced.
Often, we can rephrase our results dependent on the density \eg{by using a sparse matrix multiplication routine or relying on the Strong Triangle conjecture~\cite{abboud2014PopularConjecturesImply}}.
We prioritize readability by focusing on the endpoints of the density spectrum: $\setsize{E} \in \Oish(\setsize{V})$ and $\setsize{E} \in \Omegaish(\setsize{V}^2)$. We highlight broader density-dependent generalizations only where they are particularly salient.

Finally, we show that solving \bSMBG in $\cO((m\setsize{E})^{1-\smol}+N)$ time
for a polylogarithmic number of blocks is not possible
unless the Orthogonal Vectors Hypothesis (OVH), which is implied by SETH~\cite{williams2005new}, is refuted.
This lower bound holds for all algorithms (including non-combinatorial ones), illuminating the limits of algebraic techniques for this problem.

\begin{restatable*}{theorem}{thmLogLB}\label{thm:lb-seth}
    Unless the OVH fails, we cannot solve the \bSMBG problem in $\cO((m\setsize{E})^{1-\smol} + N)$ time for any $b=\omega(\log\setsize{V})$ and constant $\smol > 0$. 
\end{restatable*}

In a nutshell, we obtain significant improvements for graphs that have a small number of blocks ($b \leq 3$), or are sufficiently dense. 
Even for dense graphs with a rather large number of blocks, our matrix-multiplication-based algorithm is still beneficial.
Given the $\sqrt{\setsize{V}}$-block example from above, assuming that consecutive blocks are fully connected, our algorithm runs in time $\cO(\setsize{V}^{2.126} + N)$, which is faster than $\Theta(m\setsize{E} + N) = \Theta(\setsize{V}^{2.5} + N)$.

Finally, we stress that our algorithms are simple to implement, and as such, they may contribute to practical advancements in applications where the \SMLG problem is an important primitive \eg{in pangenome analyses~\cite{DBLP:journals/bioinformatics/RautiainenMM19,Rautiainen2020}}.

\subparagraph{Complexity Landscape of \hSMLG.} As long as $h$ is subpolynomial \eg{constant or logarithmic}, the above reduction from \hSMLG to \bSMBG shows that the complexity results we have for \bSMBG directly translate to \hSMLG (up to subpolynomial factors). As the results required to establish the landscape of \bSMBG require only $b \in \Theta(1)$ or $b \in \omega(\log \setsize{V})$, we also settle the landscape for \hSMLG.
Namely, there is a near-linear-time algorithm for $h = 3$; starting at $h=4$ the problem becomes hard under the BMM conjecture; and at $h \in \omega(\log \setsize{V})$ the problem becomes hard under SETH.
Also, there is a matrix-multiplication-based algorithm that is conditionally optimal for all $h$ at least 4 and at most subpolynomial. In short, while the reduction from \hSMLG to \bSMBG is efficient only for small values of $h$, surprisingly, the two problems still have the same (up to subpolynomials) complexity when parameterized by $h$ and $b$, respectively. 
This is largely due to the fact that our lower bounds show that \bSMBG (and thus \hSMLG) become hard before the reduction breaks down. After the reduction breaks down, both problems can be solved in the same $\cO(m\setsize{E} + N)$ time and not faster.

\subparagraph{Computational Model.} We assume the standard word RAM model with machine words consisting of $\Omega(\log n)$ bits, where $n$ is the input size~\cite{DBLP:books/daglib/0046305}. Basic operations on machine words, such as indirect addressing and arithmetic operations, are thus assumed to take $\cO(1)$ time.

\subparagraph{Paper Organization.} In \Cref{sec:prelims}, we provide the necessary notation and definitions as well as a few tools that we use in our algorithms. In \Cref{sec:SOTA}, we recap a version of the algorithm by Amir et al.~for solving \bSMBG. In \Cref{sec:three-blocks}, we present our near-linear-time algorithm for $b=3$ blocks (\Cref{thm:three-blocks}).
In \Cref{sec:mm-algo}, we present our algorithm employing fast matrix multiplication (\Cref{thm:mm-easy}).
Finally, in \Cref{sec:lower-bounds}, we present our conditional lower bounds for $b \geq 4$ and $b\in \omega(\log |V|)$ blocks (\Cref{thm:td-lb} and \Cref{thm:lb-seth}, respectively).

\section{Preliminaries}
\label{sec:prelims}
For all $n \in \N^+$, we set $[n] \coloneqq \{1, \dots, n\}$.
For two sets $X, Y$, we write $X \sqcup Y$ for their disjoint union.
For a third set $Z$, we also write $Z = X \sqcup Y$ to indicate that $X$ and $Y$ form a partition of $Z$.
The notation $\Oish(f(n))$ denotes $\cO(f(n) \cdot \polylog(n))$.

For a directed graph $G = (V,E)$ and a node $v \in V$, we write $\indegree{v}$ for its in-neighbors and $\outdegree{v}$ for its out-neighbors.
In the Landau notation, we also use $\indegree{v}$ and $\outdegree{v}$ for the respective sizes of these sets, that is, the in- and out-degree of node $v$.

\subparagraph{Arithmetic Progressions.}

We represent an \emph{arithmetic progression} (AP) by the tuple $(\ell, r, x) \in \N^3$; it represents the set $\{\ell + ix \mid i \in \N,\ell + ix \le r\}$.
Call $x$ its \emph{common difference}.
We will treat $(\ell,r,x)$ as a set when using the $\in$ operator (that is, when we write $i \in (\ell,r,x)$, we mean that $i$ is in the represented set, and not that $i$ is one of $\ell$, $r$, or $x$). In the word RAM model, any AP can be represented using a constant number of machine words.
The \emph{intersection} of two APs is the intersection of the sets they represent. 
In contrast, we say that two APs for sets $X$ and $Y$ \emph{overlap} if the intervals $[\min X, \max X]$ and $[\min Y, \max Y]$ intersect.
We say a set of APs is \emph{overlap-free} if its APs are pairwise non-overlapping. Similarly, we call a set of APs \emph{synchronized} if each pair of non-singleton APs is non-overlapping or shares the same common difference. Note that for singleton APs, the common difference is arbitrary \ie{they still describe the same singleton set regardless of which $x$ we choose}. Thus, for the purposes of synchronization, we disregard them.

We will use the following lemma.

\begin{restatable}{lemma}{APlemma}\label{lem:ap-intersection}
  The intersection of two APs is empty or a single AP. 
  It can be computed in $\cO(1)$ time after $\cO(m)$-time preprocessing if the largest index in any of the APs is at most $m$.
\end{restatable}

\begin{proof}
Let $A_1 = (\ell_1,r_1,x_1)$ and $A_2 = (\ell_2,r_2,x_2)$ be the two APs.
By definition, an integer $k$ lies in $A_1 \cap A_2$ if and only if $k \ge \max(\ell_1, \ell_2)$, $k \le \min(r_1, r_2)$, $k \equiv \ell_1 \pmod{x_1}$, and $k \equiv \ell_2 \pmod{x_2}$. Consider the system of congruences $k \equiv \ell_1 \pmod{x_1}$ and $k \equiv \ell_2 \pmod{x_2}$; we have the following cases.

If $\gcd(x_1, x_2) \nmid (\ell_2 - \ell_1)$, then the system has no solution, and hence $A_1 \cap A_2 = \emptyset$.
Otherwise, by the Chinese remainder theorem, the set of solutions forms a single AP with common difference $x' = \operatorname{lcm}(x_1, x_2)$ and first solution $\ell' \ge \max(\ell_1, \ell_2)$. Then $r' = \ell' + x' \cdot \lfloor (\min(r_1,r_2) - \ell') / x' \rfloor$. Note that the intersection might still be empty (and returned as such) if $\ell' > \min(r_1,r_2)$.

Every step above is implementable in $\cO(1)$ time. In particular, after $\cO(m)$-time preprocessing we can construct a data structure that allows $\cO(1)$-time $\gcd(i,j)$ queries, for any $i,j\in[m]$~\cite{DBLP:journals/jcss/KociumakaRR17}.\footnote{For our needs and simplicity, we can afford a standard implementation of the Euclidean algorithm that takes $\cO(\log m)$ time.}
\end{proof}

\subparagraph{Strings.} An \emph{alphabet} $\Sigma$ is a finite set of $\sigma$ elements called \emph{characters}.
We consider an \emph{integer} alphabet $\Sigma=[1,\sigma]$.   
For a string $S=S[1]S[2]\dots S[n]$ over $\Sigma$, we denote its length $n$ by $|S|$ and its $i$-th character by $S[i]$.
A \emph{fragment} of $S$ starting at position $i$ and ending at position $j$ of $S$ is denoted by $S[i\dd j]$.  
A \emph{prefix} of $S$ is a fragment of the form $S[1 \dd j]$, and a \emph{suffix} of $S$ is a fragment of the form $S[i \dd n]$.
A fragment $S[i\dd j]$ of $S$ corresponds to a \emph{substring} $P$ of $S$: the string composed of $S[i]\ldots S[j]$.
A substring $P$ of $S$ may have many occurrences in $S$. 
We characterize an \emph{occurrence} of $P$ in $S$ by its \emph{starting position} $i\in[n]$; that is, $P=S[i\dd i+|P|-1]$.
The set of occurrences of $P$ in $S$ is denoted by $\Occ(P,S)$. 
For two strings $S$ and $T$, we write $ST$ or $S\cdot T$ for their concatenation.

A positive integer $p$ is called a \emph{period} of a string $S$, if $S[i] = S[i + p]$, for all $i \in [1, |S| - p]$. The length of a nonempty string is a period of this string, so every nonempty string has at least one period. We define \emph{the period} of a nonempty string $S$ as the smallest of its periods, denoted by $\per(S)$.  
A \emph{border} of a nonempty string $S$ is a substring $S'$ of $S$ with $|S'|<|S|$ that is both a prefix
and a suffix of $S$. We refer to the longest border of $S$ as \emph{the border} and denote it by $\border(S)$. 
A well-known fact (cf.~\cite{DBLP:books/daglib/0020103}) is the following duality between periods and borders: for any string $S$, $\per(S)+|\border(S)|=|S|$.

\begin{lemma}[Periodicity lemma~\cite{FW1965}]
 Let $S$ be a string with periods $p$ and $q$. If
$p + q - \gcd(p, q) \leq |S|$, then $\gcd(p, q)$ is also a period of $S$.   
\end{lemma}

\begin{lemma}[\cite{DBLP:journals/siamcomp/KnuthMP77}]\label{lem:KMP}
Let $P\in\Sigma^m$.
After $\cO(m)$-time preprocessing, given $S\in\Sigma^*$,
we can compute $\Occ(P,S)$ or the 
prefixes of $P$ that are a suffix of $S$
in $\cO(|S|)$ time.
\end{lemma}

\begin{lemma}[Suffix tree~\cite{DBLP:conf/focs/Weiner73}]\label{lem:suffix}
Let $P\in\Sigma^m$. After $\cO(m\log m)$-time preprocessing,
given $S\in\Sigma^{\ell}$, we can compute one of its occurrences in $P$ (if any exists) in $\cO(\ell \log \ell)$ time.
\end{lemma}

\begin{lemma}[Suffix tree + LCP queries~\cite{DBLP:journals/tcs/LandauV86}]\label{lem:lcp}
Let $P\in\Sigma^m$. After $\cO(m\log m)$-time preprocessing, given $i$ and $j$,
we can compute the longest common prefix of $P[i \dd m]$ and $P[j \dd m]$
in $\cO(1)$ time.
\end{lemma}

\begin{lemma}[\cite{DBLP:conf/sosa/Pissis25}]\label{lem:b2}
The \bSMBG problem for $b=2$ can be solved in  $\cO(m+\setsize{E}+N)$ time.   
\end{lemma}

\section{Recap: State-of-the-Art Algorithm for \texorpdfstring{\bSMBG}{b-SMBG}}\label{sec:SOTA}

\Cref{alg:sota-overview} restates the state of the art for \bSMBG as, in many ways, it is the natural algorithm for this problem.
This is a version of the algorithm by Amir et al.~\cite{DBLP:journals/jal/AmirLL00} to specifically solve \bSMBG. We have adapted the algorithm's presentation to fit the rest of our paper; the algorithm itself is essentially unchanged. We thus omit a formal time or correctness analysis.
Our algorithms will follow the same framework, but speed up crucial operations performed by the algorithm to obtain our improvements.

\begin{algorithm}[ht]
   \textbf{Input:} Block graph $G=(V = V_1 \sqcup \dots \sqcup V_b, E, \lambda)$, pattern $P$ of length $m$
    \begin{enumerate}[noitemsep]
            \item For all $v \in V$, we maintain a set $M_v \subseteq [m]$ that represents the set of prefixes of $P$ that can be matched by a walk ending in $v$. In particular, $i\in M_v$ represents $P[1\dd i]$.
            \item For each node $v \in V_1$, we compute the borders of string $P \# \lambda(v)$, where $\#\notin\Sigma$ is a separator. We add the lengths of those borders to $M_v$.
            \item For each node $v\in V_j$, for $j\in [2,b)$, we compute the union of all the match sets of the predecessors of $v$, that is $M'_{v} \coloneqq \bigcup_{v' \in \indegree{v}} M_{v'}$. We compute all occurrences of $\lambda(v)$ in $P$. If there is an occurrence starting at position $i$ and $(i-1)\in M'_{v}$, we add $i+|\lambda(v)|-1$ to $M_{v}$.
            \item For each node $v \in V_b$, compute $M'_{v}$ as in the previous step. Compute the borders of string $\lambda(v)\# P$, where $\#\notin\Sigma$ is a separator. If there is a border of length $k$ and $(m-k)\in M'_{v}$, add $m$ to $M_{v}$.
            \item Report all nodes $v\in V_b$ such that $m \in M_{v}$. Every such node is the endpoint of a matching walk. 
    \end{enumerate}
    \caption{The state-of-the-art algorithm for \bSMBG.}
    \label{alg:sota-overview}
\end{algorithm}

\begin{theorem}[\cite{DBLP:journals/jal/AmirLL00}] \Cref{alg:sota-overview} solves the \bSMBG problem in $\cO(m\setsize{E} + N)$ time.
\end{theorem}

Let us remark that the direct reduction to standard string matching by building the concatenation of all strings given by paths of length $(b-1)$, and then searching $P$, might give $\Omega(|V|^b)$ many such strings, making this strategy unattractive for any $b\geq 2$.

\section{Exploiting Periodicity Yields Near-Linear Time for Three Blocks}\label{sec:three-blocks}
In this section, we focus on $3$-block graphs $G=(V=V_1 \sqcup V_2 \sqcup V_3, E, \lambda)$. The core bottleneck in the state-of-the-art algorithm (\Cref{alg:sota-overview}) is maintaining a set of size $m$ to represent 
how each string $\lambda(v)$, for $v\in V$, interacts with the length-$m$ pattern $P$. This means that for some middle node $v \in V_2$, we explicitly consider all occurrences of $\lambda(v)$ in $P$.
Similar considerations apply for the suffixes and prefixes of $P$ matched by the nodes of the first and last block, respectively.
Fundamentally, an approach like that cannot directly lead to an algorithm faster than $\Theta(m\setsize{V})$ as this is the size of the data just described. Our goal in this section is to circumvent exactly this bottleneck. We do this by representing how each node interacts with $P$ more efficiently.
This will then allow us to combine that information more efficiently, solving \bSMBG for $b=3$ significantly faster than the state of the art.

\subsection{How the First and Last Block Interact with the Pattern}
We take a closer look at how we can represent prefixes, suffixes, and borders of $P$. A prefix of a string is characterized by its ending position,
and a suffix of a string is characterized by its starting position.
A set of prefixes of a string can thus be characterized as a set of natural numbers, which can always be written as a union of APs.
We next prove several combinatorial properties showing that, for the cases we care about, relatively few such APs suffice and they behave nicely, in a sense that we will formalize shortly.
By symmetry, the same holds for a set of suffixes of a string.
We will use this AP representation extensively.
We similarly treat sets of borders of a string.
As a border is both a prefix and a suffix, we (choose to) identify it with the ending position of the prefix, which coincides with its length.
We start with the following well-known result that underlies our algorithmic improvement.

\begin{lemma}
  \label{lem:ap-representation}
  For any $P \in \Sigma^m$, after $\cO(m)$-time preprocessing, given $S \in \Sigma^*$, all prefixes of $P$ that are a suffix of $S$ can be computed as $\cO(\log m)$ overlap-free APs in $\cO(\setsize{S})$ time.
\end{lemma}
\begin{proof}
    We employ \Cref{lem:KMP}. The prefixes of $P$ are encoded from $[\min(m,|S|)]$, and so they can be sorted in $\cO(|S|)$ time using bucket sort.
    It is well-known that the sorted sequence of periods of a length-$n$ string can be cut into $\cO(\log n)$ (overlap-free) APs~\cite{DBLP:journals/siamcomp/KnuthMP77}. By the duality between periods and borders, this also holds for the prefixes of $P$ that are a suffix of $S$.
\end{proof}

The symmetric claim (\Cref{cor:ap-representation-sym}) holds too \eg{by reversing the strings $P$ and $S$}.

\begin{corollary}
  \label{cor:ap-representation-sym}
  For any $P \in \Sigma^m$, after $\cO(m)$-time preprocessing, given $S \in \Sigma^*$, all suffixes of $P$ that are a prefix of $S$ can be computed as $\cO(\log m)$ overlap-free APs in $\cO(\setsize{S})$ time.
\end{corollary}

We will apply \Cref{lem:ap-representation} to the pattern $P$ paired with every node label in $V_1$ (and \Cref{cor:ap-representation-sym} with every node label in $V_3$). To efficiently work with their output, we need the following structural fact about the interaction of two APs \emph{for the same $P$ but a different $S$}.

\begin{lemma}
\label{lem:long-ap-overlap}
Given $P \in \Sigma^m$ and $S_1, S_2 \in \Sigma^*$, let $\mathcal{A}_1$ be the set of lengths of all prefixes of $P$ that are suffixes of $S_1$, and $\mathcal{A}_2$ be the set of lengths of all prefixes of $P$ that are suffixes of $S_2$. Given an integer $j \ge 0$,
let $(\ell_1, r_1, x_1)$ be an AP that contains exactly all elements from $\mathcal{A}_1 \cap [2^j, 2^{j+1}) \cap [\ell_1, r_1]$  
and $(\ell_2, r_2, x_2)$ be an AP that contains  exactly all elements from $\mathcal{A}_2 \cap [2^j, 2^{j+1}) \cap [\ell_2, r_2]$.
If both APs contain at least $3$ elements, then $x_1 = x_2$. 
\end{lemma}

We will make sure that in our algorithm, the APs match the preconditions of this lemma and we can thus apply this result.
The fact that overlapping APs from different $S_1, S_2$ with at least three elements share their common difference will allow us to restore overlap-freeness (which is the nice behavior mentioned above) after combining two sets of APs.

\begin{proof}
Let $u_1, u_2, u_3$ be three consecutive elements in $(\ell_1,r_1,x_1)$ with common difference $x_1$. 
Let $v_1, v_2, v_3$ be three consecutive elements in $(\ell_2,r_2,x_2)$ with common difference $x_2$.

These elements establish $P[1 \dd u_3]$ has period $x_1$ and $P[1 \dd v_3]$ has period $x_2$. 
Since $u_1, u_3 \in [2^j, 2^{j+1})$, $u_3 - u_1 = 2x_1 < 2^j$, so $x_1 < 2^{j-1}$. Similarly, $x_2 < 2^{j-1}$.
Let $L = \min(u_3, v_3)$. Since $L \ge 2^j$ and $x_1 + x_2 < 2^j$, the prefix $P[1 \dd L]$ 
satisfies the condition of the periodicity lemma $L \ge x_1 + x_2 - \gcd(x_1, x_2)$. 
Therefore, $P[1 \dd L]$ has period $g \coloneqq \gcd(x_1, x_2)$.

Since $P[1 \dd u_3]$ has period $g$, the prefix $P[1 \dd u_3 - g]$ is a suffix of $P[1 \dd u_3]$. 
Therefore, $P[1 \dd u_3 - g]$ is a suffix of $S_1$.
We verify that $u_3 - g$ lies within the interval $\mathcal{I} = [2^j, 2^{j+1})$:
\begin{enumerate}
    \item $u_3 - g < u_3 < 2^{j+1}$ (true since $g \ge 1$).
    \item $u_3 - g \ge u_3 - x_1 = u_2 \ge 2^j$ (true since $g = \gcd(x_1, x_2) \le x_1$).
\end{enumerate}

Since $u_3 - g \in \mathcal{A}_1 \cap \mathcal{I} \cap [\ell_1,r_1]$ and the AP $(\ell_1, r_1, x_1)$ is exhaustive for that interval \ie{it contains all elements from $\mathcal{A}_1 \cap [\ell_1,r_1]$}, $u_3 - g$ must be an element of the AP.
In an AP with common difference $x_1$, if two elements $A$ and $B$ exist, then $|A-B|$ must be a multiple of $x_1$. Here, $|u_3 - (u_3 - g)| = g$. Thus,
if $g$ is a multiple of $x_1$, then $g \ge x_1$. Combined with the fact that $g = \gcd(x_1, x_2) \le x_1$, it must be that $g = x_1$.
By symmetry for the second AP, $g = x_2$. We conclude that $x_1 = x_2 = \gcd(x_1, x_2)$.
\end{proof}

In the following lemma, we show how, by iterative application of \Cref{lem:long-ap-overlap},  we combine the sets of APs for many different strings into one large set of APs that is still ``well-behaved''.

\begin{lemma}
\label{lem:ap-overlap-fix}
   Given $P \in \Sigma^m$ and $S_1, \dots, S_{k} \in \Sigma^*$, let $\{A_i\}_{i \in [k]}$ be such that each $A_i$ is the set of APs from \Cref{lem:ap-representation} for $P$ and $S_i$.
   
  Given all $A_i$ as input, we can compute a synchronized set of $\cO(k \log m)$ APs representing all elements in the union of all AP elements over all $A_i$
  in $\cO(k \log m)$ total time.
  Moreover, the non-singleton APs have at most $\cO(\log m)$ unique common differences.
\end{lemma}

\begin{proof}
The size of our input is $\cO(k \log m)$ by \Cref{lem:ap-representation}. Our only operation will be splitting APs. By this, we mean replacing one AP $(\ell, r, x)$ by two APs $(\ell_1, r_1, x)$ and $(\ell_2, r_2,x)$ that represent the same elements and have the property $\ell_1 \le r_1 < \ell_2 \le r_2$.

Let $\cA \coloneqq \bigcup_{i \in [k]} A_i$.
First, split all APs so that each resulting AP is a subset of an interval of the form $[2^j, 2^{j+1})$, for $j=0,1,\ldots,\lfloor\log m \rfloor$.
That is, split each AP at every dyadic boundary $2^j$ it intersects, so that each
resulting AP is fully contained in a single interval $[2^j,2^{j+1})$. For each $A_i$, this increases its size to at most $2\setsize{A_{i}} + \log m$ because $A_i$ was overlap-free.
Therefore, in total, the size of $\cA$ increases to at most $(2\setsize{\cA} + k\lceil \log m \rceil) \in \cO(k \log m)$.
Second, we split each AP that now contains exactly two elements into two singletons, again at most doubling the size of $\cA$.
The total size remains in $\cO(k \log m)$.
As we spent constant time per split, the time is bounded by the input and output size.

Now we want to check the property that the non-singleton APs in $\cA$ do not overlap unless they share their common difference.
Although each $A_i$ is overlap-free, APs originating from different sets
$A_i$ may overlap. To address this, we treat overlaps in two cases:

Case (1): if two APs overlap and at least one of them contains at most two
elements, then by our splitting, it is a singleton. By definition, overlaps between singletons and any other AP do not violate the synchronized definition.

Case (2): assume that two overlapping APs $(\ell_i, r_i, x_i)$, $(\ell_j, r_j, x_j)$ each have length at least three. 
By \Cref{lem:ap-representation}, both APs exhaustively cover all elements in the range $[\ell_i, r_i], [\ell_j, r_j]$ of their respective $A_i, A_j$.
Thus, we apply \Cref{lem:long-ap-overlap}, proving that $x_i = x_j$, and do nothing further.
In particular, \Cref{lem:long-ap-overlap} also shows that if any two APs of length at least three appear in the same dyadic interval, they share the same common difference, proving that there can be at most $\cO(\log m)$ unique common differences.
\end{proof}

The symmetric result holds for the APs from \Cref{cor:ap-representation-sym}.

\subsection{How to Match with the Middle Block}  
Now that we have shown how to represent the parts of the pattern matched by the first and last block and how to access that information efficiently, we use this tool for solving \bSMBG for $b=3$ faster than $\cO(m\setsize{E} + N)$ time.
Our next lemma shows that we can efficiently do the following:
Given AP positions for the prefixes of the pattern $P$ matched by the in-neighbors of a middle-block node, check which of these are followed by the label of the middle node.
In the standard \Cref{alg:sota-overview}, this step requires, for all $v^2 \in V_2$, explicitly computing all occurrences of $\lambda(v^2)$ in $P$ and then comparing this to the prefixes matched by nodes in $V_1$.
We cannot afford this as the string matching between all the labels in $V_2$ and $P$ requires $\Omega(m\setsize{V_2})$ time.
Fortunately, we can use the periodicity information about $P$ encoded in the APs representing the prefixes matched by $V_1$ to perform this important step in the algorithm much more efficiently.
The core fact underlying the following lemma is this.
\begin{observation}
\label{obs:ap-period}
    If $(\ell,r,x)$ is an AP of prefixes of $P$ that are suffixes of some string $S$ and if the AP has at least 2 elements, then $x$ is a period of $P[1\dd r]$. We define the \emph{extended periodic region} of $(\ell,r,x)$ as the longest prefix $P[1\dd r']$ of $P$, with $r'\geq r$,
    that has period $x$.
\end{observation}

The symmetric version holds for an AP of suffixes of $P$ that are prefixes of some $S$.

\begin{lemma}
  \label{lem:ap-then-match}
  We can preprocess a string $P \in \Sigma^m$ and a block graph $G=(V, E, \lambda)$ in time $\Oish(m+N_G)$ into a data structure  that can answer the following queries in $\cO(1)$ time:
  
  Given an AP $(\ell, r, x)$, representing the prefixes of $P$ that are a suffix of a label $\lambda(v^1)$, and some (other) label $\lambda(v^2)$, output all occurrences $i$ of $\lambda(v^2)$ in $P$ such that $i-1 \in (\ell, r, x)$.
  These occurrences are output as $\cO(1)$ APs, each of which is a sub-AP with common difference $x$ of the input AP $(\ell, r, x)$.
  Also, any occurrence of $\lambda(v^2)$ that would end after the extended periodic region is output as a singleton AP.
\end{lemma}

\begin{proof}
As preprocessing, we construct the LCP data structure from \Cref{lem:lcp} on $P$. In addition, for each node label $\lambda(v)$ we use \Cref{lem:suffix} to compute one occurrence of $\lambda(v)$ in $P$, if such an occurrence exists. The resulting position is stored for $\cO(1)$-time access. 
This information enables us to verify (other) occurrences of $\lambda(v)$ in $P$ in $\cO(1)$ time
by asking LCP queries on $P$.
The total preprocessing time is thus $\Oish(m+N_G)$.

\begin{figure}[tbhp]
    \centering
    \includegraphics[page=7]{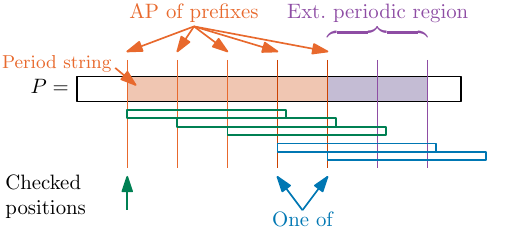}
    \caption{Illustration of \Cref{lem:ap-then-match}. Orange lines depict an AP of prefixes of $P$.
    As purple, further repetitions of the \emph{period string}.
    Below, all possible occurrences of a string $\lambda(v^2)$ directly \emph{after} an AP element. These fall into the two categories used in the proof: entirely within the periodic region (green), and partially or entirely after (blue). For green, we only need to check a single starting position (indicated by an arrow) to deduce all occurrences. For blue, we only need to check one of the starting positions in the indicated range.}
    \label{fig:ap-then-match}
\end{figure}

Now consider a query consisting of an arithmetic progression $(\ell,r,x)$ and a label $\lambda(v^2)$.
Assume $r$ is the AP's largest element and let $\alpha \coloneqq \lfloor(r-\ell)/x\rfloor +1$ denote the number of its elements. 
Let $\alpha' \ge \alpha$ be the maximum integer such that we have $\alpha'$ full repetitions in the extended periodic region (which is defined as long as $\alpha\geq 1$), that is, the fragments $P[(\ell+t x +1) \dd (\ell + (t+1)x)]$ are equal
for every $t=0,1,\ldots,\alpha'-1$, 
but $P[(\ell+\alpha' x +1) \dd (\ell + (\alpha'+1)x)]$
is either different or undefined. If $\alpha\le 2$, then these positions can be checked directly using LCP queries (\Cref{lem:lcp}). Henceforth, assume that $\alpha\ge 3$. Using one LCP query, we can check how far this periodic region extends to the right (computing $\alpha'$).
We distinguish possible occurrences of $\lambda(v^2)$ into two categories; see \Cref{fig:ap-then-match}.

First (green in the figure), consider possible occurrences of $\lambda(v^2)$ that start at or before $(r + 1)$ and are fully contained within the periodic region. We check whether $\lambda(v^2)$ occurs at position $(\ell + 1)$. If so, then by periodicity, it also occurs starting at every AP position $\ell + t x + 1$ as long as the occurrence is fully contained in the periodic region $P[\ell+1 \dd \ell+\alpha' x]$.
These occurrences therefore form a single AP.

Second (blue in the figure), consider possible occurrences of $\lambda(v^2)$ that start at or before $(r + 1)$ and end after the periodic region.
If such an occurrence succeeds the $i$-th-before-last AP position \ie{$(r+1)$ succeeds $i=0$ and so on}, then $\lambda(v^2)$ must start with $\min(i+\alpha'-\alpha+1, \lfloor \setsize{\lambda(v^2)}/x \rfloor)$ repetitions of the period string.
By computing the maximum number of period-string repetitions at the start of $\lambda(v^2)$ using an LCP query, we can deduce $i$ from $\alpha$, $\alpha'$, and $\setsize{\lambda(v^2)}$.
Knowing $i$, we verify the occurrence using a single additional LCP query.
Since there are exactly $\alpha'$ repetitions of the period string, 
there is at most one such alignment $i$.
Thus, we add either nothing or one singleton AP to the output.

Each case contributes at most one AP and all checks take $\cO(1)$ time after preprocessing.
\end{proof}

\subsection{Putting Everything Together}
Now that we have shown how to combine multiple sets of APs into a single well-behaved set, we define the intersection operation.
We provide an intuition right after the statement.

\begin{lemma}
  \label{lem:ap-two-set-intersection}
  Let $P \in \Sigma^m$, and let $S_1, \dots, S_k \in \Sigma^*$ with corresponding AP sets $A^1_i$ that come from applying \Cref{lem:ap-then-match} to \Cref{lem:ap-representation} for some $v^2 \in V_2$.
  Also, let $S'_1, \dots, S'_{k'} \in \Sigma^*$ with corresponding AP sets $A^3_j$ from \Cref{cor:ap-representation-sym}.
  Given all $A^1_i$, all $A^3_j$ and $v^2$ as input, we can compute in $\cO((k+k') \log^{2} m)$ time if there is a value $z$ such that $z$ is in some AP in some $A^1_i$ and $z+\setsize{\lambda(v^2)}+1$ is in some AP in some $A^3_j$. 
\end{lemma}

One should think about the strings $S_1, \dots, S_k$ as node labels from the first block of $G$, and the strings $S_1', \dots, S'_{k'}$ as node labels from the last block of $G$.
This middle-block node is connected to the respective nodes from the first and last block.
We may also assume that the label of the middle node appears in $P$ directly after every position listed in all $A^1_i$ by \Cref{lem:ap-then-match}.
We now want to determine the positions in $P$ where this is completed into a full match of $P$ by the connected first and last block nodes.

\begin{proof}
We may assume \Cref{lem:ap-overlap-fix} had been applied to $\{A^1_i\}_{i \in [k]}$ and the symmetric version of the lemma to $\{A^3_j\}_{j \in [k']}$.
Let
\(
\cA^1 \coloneqq \bigcup_{i \in [k]} A^1_i\) and \(
\cA^3 \coloneqq \bigcup_{j \in [k']} A^3_j.
\)
Next, shift all APs in $\cA^3$ by $(-\setsize{\lambda(v^2)}-1)$; that is, decrease both the lower and upper bounds of each AP by $\setsize{\lambda(v^2)}+1$. To distinguish between the shifted and unshifted version, we add a $\hat{}$ symbol on the shifted version. 
This preserves the common difference
and the number of APs, and can be done in $\cO(|\cA^3|) = \cO(k' \log m)$
time.
After this transformation, the desired output is
\(
\bigcup \cA^1 \;\cap\; \bigcup \hat{\cA}^3.
\)
To determine if this intersection is empty efficiently, proceed as follows.

We extract as singletons the last element of every AP in $\cA^1$ and the first element of each AP in $\hat \cA^3$. If, after this, any AP has fewer than three elements, 
split it into singleton APs.
This increases the number of APs by a constant factor.
Sort the singleton APs by value.
Since $|\cA^1| + |\hat{\cA}^3| = \cO((k+k') \log m)$, sorting
takes $\cO((k+k') \log^2 m)$ time.
In this sorted list, check if there are two singleton APs that intersect.

Afterwards, we do a second pass in which we check whether a singleton AP of $\cA^1$ intersects a non-singleton AP of $\hat \cA^3$. We then repeat the process with switched roles of $\cA^1$ and $\hat{\cA}^3$.
We do a sorted scan of the APs and maintain a balanced search tree (BST)~\cite{DBLP:books/daglib/0046305} with the currently active non-singleton APs $(\hat \ell_3, \hat r_3,x_3)$ \ie{those with active interval $[\hat \ell_3, \hat r_3]$}, indexed by $\hat{\ell}_3 \bmod x_3$. 
If there are collisions in the BST, we only need to keep the AP that extends the furthest to the right.
Since all non-singleton APs are synchronized by \Cref{lem:ap-overlap-fix}, all APs active at the same time have the same common difference $x_3$. 
When scanning a singleton AP, consisting of a single element $\ell_1$, we query the BST with the value for $\ell_1 \bmod x_3$.
Observe that $\ell_1$ intersects an active AP $(\hat \ell_3,\hat r_3, x_3)$ if and only if $\ell_1 \equiv \hat \ell_3 \pmod{x_3}$, thus we find an intersection if one exists.
Since all keys are in $\cO(m)$, BST operations take $\cO(\log m)$ time, the sorting
takes $\cO((k+k') \log^2 m)$ time, and so the total time is $\cO((k+k')\log^2 m)$.

We are left to argue that if there is an intersection between two non-singleton APs, then there is an intersection involving a singleton AP.
For this assume that $a^1 = (\ell_1, r_1, x_1)$ from $\cA^1$ and $\hat a^3 = (\hat \ell_3, \hat r_3, x_3)$ from $\hat \cA^3$ are non-singleton APs that intersect.
If $a^1$ and $\hat a^3$ intersect, they must overlap.
Also, since we extracted the last element of $a^1$ and first element of $\hat a^3$, the original APs overlapped by at least $x_1 + x_3 + 1$.
Let $r'_1$ be the last element of $a^1$ before the extraction.
We know the extended periodic region of $a^1$ extends at least $\setsize{\lambda(v^2)}$ to the right after $r'_1$ as otherwise $r'_1$ would have been output as a singleton by \Cref{lem:ap-then-match}.
Thus, the extended periodic region of $a^1$ overlaps with $\hat a^3$ by at least $x_1+x_3+\setsize{\lambda(v^2)}+1$.
After undoing the shift, the extended periodic regions overlap by at least $x_1 + x_3$.
To apply the periodicity lemma, we now also need that $x_1$ and $x_3$ are periods of this overlap.
This follows from \Cref{obs:ap-period}, which we can apply since all APs we have are cut up APs of those originally produced using \Cref{lem:ap-representation}.
In conclusion, $\gcd(x_1,x_3)$ is a period of the overlap of the extended periodic regions and thus of $P$.

Now, let $y$ be an element of the intersection of $a^1$ and $\hat a^3$, then any $y+z\cdot \gcd(x_1, x_3)$ must be in $\hat{\cA}^3$.
This holds since if $P[s\dd m]$ is a suffix of $P$ that is a prefix of some $S'_j$, the same is true for $P[s + z \cdot p\dd m]$, where $p$ is any period of $P[s\dd m]$, so in particular for $s = y+\setsize{\lambda(v^2)}+1$ and $p = \gcd(x_1,x_3)$.
Clearly, $r'_1$ can be written as $\ell_1 + z \cdot x_1$ and thus as $y+z'\cdot \gcd(x_1, x_3)$. Therefore, a singleton from $\cA^1$ intersects an AP from $\hat \cA^3$.
\end{proof}

Together, \Cref{lem:ap-two-set-intersection,lem:ap-then-match} enable the core matching step of our approach.
Our improved algorithm (stated as \Cref{alg:ap-overview}) retains the overall structure of \Cref{alg:sota-overview} for $b=3$, but leverages the data structures and combinatorial insights developed in this section to carry out these core steps significantly more efficiently.

\begin{algorithm}[ht]
    \textbf{Input:} Block graph $G=(V_1 \sqcup V_2 \sqcup V_3, E, \lambda)$, pattern $P$ of length $m$
    \begin{enumerate}
        \item Preprocess $P$ and $G$ using \Cref{lem:ap-then-match}. Preprocess $P$ using \Cref{lem:ap-representation,cor:ap-representation-sym}.
        \item For each $v^1 \in V_1$, compute $A_{v^1}$ using \Cref{lem:ap-representation} on $P$ and $\lambda(v^1)$.
        \item For each $v^3 \in V_3$, compute $A_{v^3}$ using \Cref{cor:ap-representation-sym} on $P$ and $\lambda(v^3)$.
        \item For each $v^2 \in V_2$:\label{step:for-loop}
        \begin{enumerate}[noitemsep,label=(\alph*)]
            \item Let $\cA^1 \coloneqq  \bigcup_{v^1 \in \indegree{v^2}} A_{v^1}$ represent the set of prefixes matched by in-neighbors of $v^2$.\label{step:left-match}
            \item Let $\cA^3 \coloneqq  \bigcup_{v^3 \in \outdegree{v^2}} A_{v^3}$ represent the set of suffixes matched by out-neighbors of $v^2$.
            \item For each AP from $\cA^1$, use the query from \Cref{lem:ap-then-match} to restrict to the prefixes of $P$ that are directly followed by an occurrence of $\lambda(v^2)$.\label{step:ap-the-match}
            \item Apply \Cref{lem:ap-two-set-intersection} to (the restricted) $\cA^1$ and to $\cA^3$ to compute if there is a position in $P$ where an occurrence of $\lambda(v^2)$ can be extended to a full match of $P$ in $G$.
        \end{enumerate}
    \end{enumerate}
    \caption{Our improved algorithm for \bSMBG and $b=3$.}
    \label{alg:ap-overview}
\end{algorithm}

\thmThreeBlocks

\begin{proof}
  For \bSMBG, we analyze \Cref{alg:ap-overview}.
  The implication for \hSMLG for $h=3$ follows from the reduction outlined in the introduction.

  First, observe that in Step~\ref{step:for-loop}\ref{step:ap-the-match}, every AP queried using \Cref{lem:ap-then-match} is a subset of an AP computed in Step~\ref{step:for-loop}\ref{step:left-match}.  Hence, all elements of these APs are prefixes of $P$ that are suffixes of some label from the first block, satisfying the precondition of \Cref{lem:ap-then-match}.

    The running time follows by simply adding the running times from \Cref{lem:ap-representation,cor:ap-representation-sym,lem:ap-overlap-fix,lem:ap-two-set-intersection,lem:ap-then-match}.
  For correctness, observe that \Cref{alg:ap-overview} computes exactly the same sets of indices as the \Cref{alg:sota-overview} does, just that here we represent them as APs.
\end{proof}

\subparagraph{Reporting All Occurrences.} We extend \Cref{alg:ap-overview} to report all occurrences of $P$ in $G$ in the form $(v_1, v_2, v_3, i_1, i_3)$, where $P = \lambda(v_1)[i_1\dd |\lambda(v_1)|] \cdot \lambda(v_2) \cdot \lambda(v_3)[1\dd i_3]$, making it more useful for the applications outlined in the introduction (as opposed to a purely decision algorithm).
In the proof of \Cref{lem:ap-two-set-intersection}, we already find all occurrences that are the result from the intersection of two singleton APs.
To find intersections between singletons and non-singleton APs, we modify the algorithm to no longer discard collisions in the BST.
This incurs a $+\mathrm{output}$ overhead (think of a graph with many nodes with identical labels in the first and last block).
To also get all occurrences for two non-singleton APs $a^1,\hat a^3$, recall that we proved that there must be some corresponding singleton $r_1'$ from $\cA^1$ that intersects an AP from $\hat \cA^3$ (this, we find in the other cases).
From our correctness argument, it follows that if, for each match involving at least one singleton, we also compute the intersection of their corresponding original, unsplit APs, we are guaranteed to find all occurrences. We use \Cref{lem:ap-intersection} to intersect non-singleton APs in $\cO(1)$ time. Adding the required bookkeeping to make this association with the original versions of the APs is a bit tedious, but not too difficult.
In particular, we will take care to find the $\cA^3$ AP as originally computed using \Cref{cor:ap-representation-sym} (we then repeat the shift) and the $\cO(1)$ APs that result from applying \Cref{lem:ap-then-match} to the original $\cA^1$.

\subparagraph{Why not $b\geq 4$?} A natural question arising is whether the same scheme can be extended to more than three blocks.
Since \Cref{alg:sota-overview} works for an arbitrary number of blocks and \Cref{alg:ap-overview} can be seen as an AP-based implementation of the same general scheme, we can also derive a correct AP-based algorithm for an arbitrary number of blocks.
However, we only get the additional structure that bounds the size of the AP-based representation for the part of the pattern matched by the first and last block.
We make the idea work for three blocks by matching the middle block based on the relatively few positions we need to check after computing where the first and last block align to even allow a match in the middle block.
For $b = 4$, this idea does not work anymore.

This intuition is made formal by the lower bounds we show in \Cref{sec:lower-bounds}, where we show that for four blocks it is conditionally impossible to have a near-linear time algorithm (see \Cref{thm:td-lb}).

 \subparagraph{Solving \SMBG for $b=3$.}
 To solve \SMBG, the more general problem where matches are not required to span all $b=3$ blocks, we do the following.
 Let $G=(V=V_1 \sqcup V_2 \sqcup V_3, E, \lambda)$ and $P$ be the input to \SMBG.
We apply \Cref{lem:KMP} using $P$ and all the node labels $\lambda(v)$ for $v\in V_1 \sqcup V_2 \sqcup V_3$, such that $|\lambda(v)|\geq m $. The total time is $\cO(m+N)$.
We then apply \Cref{lem:b2} two times: on the graph induced by $V_1\sqcup V_2$ from $G$; and on the graph induced by $V_2\sqcup V_3$ from $G$. 
The total time is $\cO(m+\setsize{E}+N)$.
Finally, we apply \Cref{thm:three-blocks} on $G$.
The time is $\Oish(m + \setsize{E} + N)$.
We have arrived at the following result.

\begin{corollary}\label{coro:three-blocks}
The \SMBG problem for $b=3$ can be solved in $\Oish(m + \setsize{E} + N)$ time. 
\end{corollary}

\section{An Algorithm Based on Matrix Multiplication}\label{sec:mm-algo}
In this section, we give an alternative algorithm for the \bSMBG problem that works for an arbitrary number $b$ of blocks.

When solving the \bSMBG problem, a core challenge arises. For each node $v$ in our block graph $G$, we need to merge the occurrences of $\lambda(v)$ in $P$ with the prefixes of $P$ matched by the predecessors of $v$ in $G$. \Cref{alg:sota-overview} considers all predecessors of $v$ in $G$ individually, introducing the $\cO(\setsize{E})$ factor to the algorithm's running time.
We avoid performing this costly operation individually for each node and compute them for a complete block simultaneously using matrix multiplication.

We use $\omega \in [2, 2.372)$ as the running time exponent of matrix multiplication of square matrices~\cite{DBLP:conf/soda/WilliamsXXZ24}. We will also make use of the notation $\omega(a,b,c)$ to denote the running time exponent of rectangular matrix multiplication of an $n^a \times n^b$ by an $n^b \times n^c$ matrix, and we use~\cite{Complexity} to calculate the concrete value of this function.
We first state the algorithmic result which we prove in this section.

\begin{theorem}
  \label{thm:mm-generality}
  The \bSMBG problem can be solved in  $\cO\left(\left(\sum_{i \in [b-1]} n^{\omega(x, y_i, y_{i+1})}\right) +N\right)$ time, where $n^x = m$ and $n^{y_i} = \setsize{V_i}$, for each $i \in [b]$.
\end{theorem}

The algorithm is given in \Cref{alg:fmm-overview}. For easier understanding, we give the following looser bound.

\thmMMEasy

\begin{proof}
    Let $n \coloneqq \max(\setsize{V}, m)$. We may assume that $m \geq \setsize{V}$, else we increase $m$.
    For $b \in \cO(1)$ blocks, the statement follows directly. For more blocks, 
    we need to show that $\sum_{j \in [b-1]}n^{\omega(1,y_{j}, y_{j+1})} \in \cO(n^\omega)$ where $n^{y_j} = |V_j|$ for all $j \in [b]$.
    Assuming $y_j \leq y_{j+1}$, note that we can decompose a multiplication of rectangular matrices with size $n \times n^{y_j}$ and $n^{y_j} \times n^{y_{j+1}}$ into $n/n^{y_j} \cdot n^{y_{j+1}}/n^{y_j}$ square matrix multiplications of size $n^{y_j} \times n^{y_j}$. The remaining proof follows by convexity.
    Note that we have for all $j \in [b]$ that $y_j \leq 1$ and $\sum_{j \in [b]} n^{y_j} = \setsize{V}$ ($\star$). For all $j \in [b-1]$, we set $y^*_j = \max(y_j, y_{j+1})$.
    \begin{align*}
        \sum_{j \in [b-1]}n^{\omega(1, y_j, y_{j+1})}
        & <\sum_{j \in [b-1]}n^{\omega(1, y^*_j, y^*_j)} \\
        & \leq \sum_{j \in [b-1]} \frac{n}{n^{y^*_j}} (n^{y^*_j})^\omega\\
        & \leq n\sum_{j \in [b-1]} n^{(\omega-1) y^*_j} \\
        & < 2n \sum_{j \in [b]} n^{(\omega-1)y_j} \tag{convexity, ($\star$)} \\
        & \leq 2n^{\omega}
    \end{align*}
\end{proof}

\Cref{alg:fmm-overview} iteratively computes the matching prefix of the pattern $P$ to the first $i$ blocks of the input graph. Inside the algorithm, we use the matrix $\resultMat{j} \in \{0,1\}^{m \times |V_j|}$ for which $\resultMat{j}[i,v] = 1$ if and only if there exists a path in $V_1 \times \ldots \times V_j$ ending in $v$ that matches the prefix $P[1\dd i]$ of $P$ (as a suffix).
The matrix $\matchMat{j} \in \{0,1\}^{m \times |V_j|}$ at cell $(i,v)$ indicates whether the label $\lambda(v)$  matches $P$ at $P[i\dd i+|\lambda(v)| - 1]$.
The intermediate matrix $\interMat{j} \in \{0,1\}^{m \times |V_j|}$ stores for each prefix $P[1\dd i]$ and each node $v$ whether a predecessor of $v$ already matches $P[1\dd i]$. 
We denote binary matrix multiplication using the notation $(\cdot)$.

\begin{algorithm}
    \textbf{Input:} Block graph $G=(V = V_1 \sqcup \dots \sqcup V_b, E, \lambda)$, pattern $P$ of length $m$
    \begin{enumerate}[noitemsep]

    \item For every node $v$ in $V_1$, compute a binary column vector $\matchVec{v}$ of length $m$ s.t. $\matchVec{v}[i] = 1$ if $P[1\dd i]$ is a suffix of $\lambda(v)$. Assemble these into a $\{0,1\}^{m \times |V_1|}$ matrix $\matchMat{1}$. \label{step:one}
    
    \item Construct the adjacency matrix $\adjMat{1}{2}$ between $V_1$ and $V_2$. Compute $\interMat{2} \coloneqq  \matchMat{1} \cdot \adjMat{1}{2}$. \label{step:two} 

    \item For every node $u$ in $V_2$, compute a binary column vector $\matchVec{u}$ of length $m$ s.t. $\matchVec{u}[i] = 1$ if $i \in \Occ(\lambda(u),P)$. Assemble these into a match matrix $\matchMat{2} \in \{0,1\}^{m \times |V_2|}$. \label{step:middle-check}
    
    \item Compute the matrix $\resultMat{2}$ of prefix matches of $P$ by $V_1$ and $V_2$. Set
    $\resultMat{2}[i, u] \coloneqq \interMat{2}[i', u] \land \matchMat{2}[i' + 1, u]$ for all $i \in [m], u \in V_2$ with $i' = i - |\lambda(u)|$. If $i' \leq 0$, the matrix value is treated as $0$. \label{step:compose}

    \item Repeat Steps \ref{step:two} to \ref{step:compose} for all blocks $j \in [3, b)$, using $\resultMat{j-1}$ in place of $\matchMat{1}$.
    
    \item For the last block $b$, similar to the first, construct $\matchMat{b} \in \{0,1\}^{m \times |V_b|}$ such that $\matchMat{b}[i, v] = 1$ for any $v \in V_b$ if $P[i \dd m]$ is a prefix of $\lambda(v)$.
    Repeat Steps \ref{step:two} and \ref{step:compose}, using this $\matchMat{b}$ instead in Step \ref{step:compose}.
    \label{step:pattern-ending-check}

    \item If $\resultMat{b}$ contains at least one $1$, we have found a matching walk $v_1,\ldots,v_b$.  \label{step:match}
\end{enumerate}
    \caption{Our algorithm for \bSMBG based on matrix multiplication.}
    \label{alg:fmm-overview}
\end{algorithm}

\begin{proof}[Proof of \Cref{thm:mm-generality}]
    Towards correctness, we notice that we compute the same sets of indices as in \Cref{alg:sota-overview}.
    Assume that a match needs to span all $b$ blocks of the graph and all matrices $\matchMat{j-1}, \interMat{j}, \resultMat{j}$ for $j \in [2, b]$ are computed as stated.
    For any $u \in V_{j}, i \in [m]$, the value $\interMat{j}[i, u]$ indicates that $P[1\dd i]$ is a suffix of the combined string of a path ending in a predecessor of $u$.
    The matrix $\resultMat{j}[i, u] = 1$ if and only if both, the prefix $P[1 \dd i - |\lambda(u)|]$ of $P$ is a suffix of a path ending in a predecessor of $u$, as given by $\interMat{j}$, and the label of $u$ lies in the correct place in the pattern, i.e., $P[i - |\lambda(u)| + 1 \dd i] = \lambda(u)$, as given by $\matchMat{j}$.

    Note that the special handling of blocks $V_1, V_b$ is identical to \Cref{alg:sota-overview}, since it suffices that $P$ is a suffix (prefix) of a node in $V_1$ ($V_b$), the computation of the corresponding match matrices is correct.

    For the running time, \Cref{lem:KMP} gives the computation of Steps \ref{step:one} and \ref{step:pattern-ending-check} in time $\cO(m(|V_1| + |V_b|))$. By \Cref{lem:KMP}, Step \ref{step:middle-check} runs in time $\cO(N+m\setsize{V})$ over the course of the algorithm.
    Steps \ref{step:compose} and \ref{step:match} are constant-time look-ups for each cell of the $m \times |V_j|$ matrix.
    Step \ref{step:two} dominates the running time with the rectangular (binary) matrix multiplication of a $m \times |V_{j-1}|$ with a $|V_{j-1}| \times |V_j|$ matrix, giving a running time of $n^{\omega(x, y_{j-1}, y_{j})}$ for $n^x = m$, $n^{y_{j-1}} = |V_{j-1}|$, $n^{y_{j}} = |V_{j}|$. 
\end{proof}

Note that if our block graph is sufficiently sparse and our pattern sufficiently small with $m\setsize{E} \ll \max(\setsize{V},m)^{\omega}$, \Cref{alg:sota-overview} is faster than \Cref{thm:mm-generality}. However, using sparse (rectangular) matrix multiplication instead~\cite{gustavson1978,abboud2024TimeComplexityFully}, we always match or improve upon the $\cO(m\setsize{E})$ running time.

\subparagraph{Reporting All Occurrences.}
\Cref{alg:fmm-overview} decides whether an occurrence exists. To output any one such walk, we trace back which values of $\interMat{j}$ and thus $\resultMat{j-1}$ were non-zero, thus constructing an occurrence.

Furthermore, \Cref{alg:fmm-overview} allows us to count the number of occurrences. 
Note that the intermediate matrix $\interMat{j}$ does not need to be binary. 
Using the ordinary matrix product instead of the binary matrix multiplication, 
it instead computes the number of prefix block matches for each node. Further using multiplication instead of a logical operation in Step (4) yields the number of prefix matches of $P$ for $\resultMat{b}$. 

\subparagraph{Solving \SMBG.} In fact, with some extra care, the algorithm underlying \Cref{thm:mm-generality}, solves the more general \SMBG problem where matches must not span all blocks within the same time complexity.

\begin{corollary}
  \label{coro:mm-generality}
  The \SMBG problem can be solved in  $\cO\left(\left(\sum_{i \in [b-1]} n^{\omega(x, y_i, y_{i+1})}\right) +N\right)$ time, where $n^x = m$ and $n^{y_i} = \setsize{V_i}$, for each $i \in [b]$.
\end{corollary}
\begin{proof}
    In the proof of \Cref{thm:mm-generality}, we have assumed that the occurrences of $P$ span all $b$ blocks of the graph.
    To get rid of this assumption, we may simultaneously assume that a block $V_j$ is either the first, the last, or some regular intermediate block. 
    To handle the pattern starting in block $V_j$, we simultaneously compute the matching matrix $\matchMat{j}$ according to Step (3), detecting internal occurrences of a node's label in $P$, and according to Step (1), matching a prefix of $P$ to a suffix of a node's label. 
    When combining the prefix block matching of predecessors in Step (4), we also set $\resultMat{j}[i,u] = 1$ if $P[1\dd i]$ is a suffix of $\lambda(u)$.  
    To handle the pattern ending in block $V_j$, we simultaneously run Steps (6) and (7) for each block. Both of these extra measures introduce only a constant factor overhead to the running time of \Cref{thm:mm-generality}. 
\end{proof}

In the following section, we establish that \Cref{thm:mm-generality} is optimal (up to subpolynomial factors) if the number of blocks $b$ is at least four for sufficiently dense block graphs.

\section{Lower Bounds}\label{sec:lower-bounds}

Comparing our results for $b \leq 3$ and $b \geq 4$ blocks, we see a jump in the time complexity.
While for $b \leq 3$ we have near-linear-time algorithms (\cite{DBLP:conf/sosa/Pissis25} and \Cref{alg:ap-overview}), we only have an algebraic $\cO(\max(\setsize{V}, m)^\omega + N)$-time or a combinatorial $\cO(m\setsize{E} + N)$-time algorithm for $b \geq 4$ (\Cref{thm:mm-easy} and~\cite{DBLP:journals/jal/AmirLL00}), which for $\omega > 2$ represents a significant increase.

We show that these algorithms are conditionally optimal by giving a matching lower bound, implying that the observed difference in complexity is inherent to \bSMBG.
We further give a second lower bound conditioned on OVH showing that the $\cO(m\setsize{E} + N)$-time algorithm is conditionally optimal already for polylogarithmically many blocks.
Both lower bounds already hold for alphabets of size at least $3$ (so in particular, for the DNA alphabet).

We start in \Cref{sec:hardness} with the hardness assumptions that we rely on for our conditional lower bounds. We then prove \Cref{thm:td-lb,thm:lb-seth} in \Cref{sec:TD,sec:OV}, respectively.

\subsection{Hardness Assumptions}\label{sec:hardness}

Let us consider the well-known $k$-SAT problem.
Given a propositional logic formula in conjunctive normal form which has at most $k$ literals in each clause,
determine whether there exists an interpretation that satisfies the input formula.
\begin{seth} For any constant $\smol > 0$, there exists $k \geq 3$ such that the $k$-SAT problem cannot be solved in $\cO(2^{(1-\smol)n})$ time.
  \end{seth}
The Orthogonal Vector (OV) problem is defined as follows.
Given two sets of $n$ vectors $A,B \subset \{ 0,1\}^d$,
determine whether $\exists \alpha \in A, \exists \beta \in B: \langle \alpha, \beta \rangle = \sum^d_{i=1}{\alpha[i]\beta[i]} = 0$.
It is well-known that SETH implies the following hypothesis~\cite{williams2005new}.
\begin{ovh}
    For $d \in \omega(\log n)$, there is no constant $\smol > 0$ such that the OV problem can be solved in $\cO(n^{2-\smol})$ time.\footnote{For simplicity, we use $d \in \omega(\log n)$. However, a slightly weaker assumption is already implied by SETH, namely: for any constant $\smol > 0$, there is a $c > 1$ such that the OV problem cannot be solved in time $\cO(n^{2-\smol})$ for $d = c \log n$~\cite{williams2005new}.}
\end{ovh}
We further use the Triangle Detection (TD) problem, which consists of finding a triangle in an (undirected) graph or concluding that there is no triangle in the graph. 
\begin{tdh}
        Let $\smol > 0$ be a constant. There is no combinatorial algorithm that solves the TD problem on $n$-node graphs in $\cO(n^{3-\smol})$ time. Likewise, there is no (algebraic) algorithm which solves the TD problem in $\cO(n^{\omega-\smol})$ time.
\end{tdh}
The Triangle Detection hypothesis is a special case of the $k$-Clique hypothesis (for $k =3$) and is fine-grained equivalent to the well-studied and long-standing (combinatorial) BMM conjecture~\cite{roditty2011DynamicShortestPaths,abboud2014PopularConjecturesImply}, stating that (combinatorial) matrix multiplication over the Boolean semiring cannot be solved faster than the standard (combinatorial) matrix multiplication over integers, see \cite{williams2019FineGrainedQuestionsAlgorithms} for an overview.
Note that an $\cO(n^\omega)$-time and a combinatorial $\cO(n^3)$-time algorithm are known for the TD problem~\cite{DBLP:journals/siamcomp/ItaiR78}. The definition of \emph{combinatorial} is a subject of debate. For our purposes, it suffices to forbid the use of algebraic tools such as fast matrix multiplication or Fast Fourier Transform.\footnote{While there is some progress towards formalizing the notion, another common view in the fine-grained complexity community is that we should rather see the conjecture as a tool to show that problems are hitting a barrier in our current toolkit. 
In essence, for these problems, reducing the time complexity below a certain threshold acts as a direct indicator that the problem has been reduced to fast matrix multiplication.
Regardless of the precise meaning of the conjecture, showing that TD reduces to a problem shows that it is among a set of hundreds of problems where nobody knows an algorithm polynomially faster than $\cO(n^3)$ except for Strassen's algorithm and its successors. It also shows that solving this problem faster with a different algorithm would also immediately give a different truly subcubic algorithm for BMM. This has eluded us for a long time. See Section 1.1 of \cite{DBLP:conf/stoc/AbboudFKLM24} for a longer discussion about this topic and why we should care about combinatorial algorithms in the first place.} 

\subsection{Constant Number of Blocks}\label{sec:TD}
Our first lower bound applies to the case where $b \geq 4$ by a reduction 
from the TD problem. The corresponding definition and hardness assumption can be found in \Cref{sec:hardness}. In summary, the problem of finding a triangle in a graph is fine-grained equivalent to BMM.

\begin{theorem}
    \label{thm:td-lb}
    If we can solve \bSMBG with $b \geq 4$ blocks on a binary alphabet in time $T_4(\setsize{V}, N, m)$, we can solve TD in a graph with $n$ nodes in time $\cO(T_4(n, n^2, n) + n^2)$. 
\end{theorem}

\begin{proof}
Let $G'=(V',E')$, with $\setsize{V'} = n$, be the graph in which a triangle is to be detected. Assume that $G'$ does not contain self-loops.
We define a block graph $G_B = (V, E)$ with $b=4$ blocks $V = V_1 \sqcup \ldots \sqcup V_4$, each containing a copy of the nodes $V'$. 
We number the nodes $V'=\{v_1, \dots, v_n\}$ and write $v^j_i$ for the copy of $v_i$ in $V_j$.
For each $i \in [n]$, we set the label of each node to
\begin{align*}
    &\ell(v_i^1) \coloneqq 10^{i}, \\ 
    &\ell(v_i^2) = \ell(v_i^3) \coloneqq 0, \text{ and} \\ 
    &\ell(v_i^4) \coloneqq 0^{n-i}1. 
\end{align*}
The edges between blocks $V_j, V_{j+1}$, for $j \in [3]$, correspond to the edges from $E'$, that is, for each $(v_{i_1},v_{i_2}) \in E'$ we add the edges $(v^j_{i_1}, v^{j+1}_{i_2})$ to $E$ for each $j \in [3]$.
We set the pattern to be $P \coloneqq 1 \cdot 0^{n+2} \cdot 1$.
See \Cref{fig:td-lb-dup} for an illustration.

\begin{figure}[ht]
    \centering
    \includegraphics[width=0.65\linewidth,page=1]{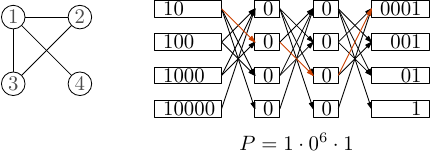}
    \caption{Left, an undirected graph with $n=4$. Right, the corresponding construction of the block graph from the proof of \Cref{thm:td-lb}. One block match for the triangle $123$ is indicated in red.
    Intuitively, we find triangles like this as all paths through the graph that start and end at copies of the same node will have $1\,0^{n+2}\, 1$ as their concatenation of labels. }
    \label{fig:td-lb-dup}
\end{figure}

We have a match in $G_B$ if there are four nodes \[(v^1_{i_1}, v^2_{i_2}, v^3_{i_3}, v^4_{i_4}) \in V_1 \times V_2 \times V_3 \times V_4\] such that: (1) every two consecutive nodes share an edge; and (2) $v_{i_1}$ and $v_{i_4}$ are copies of the same node from $V'$. As $G'$ does not have self-loops, all other nodes cannot be copies of the same node in $V'$.
In that case, the incidences imply a triangle in $G'$.
For the reverse direction, a triangle $C \in (V')^3$ in $G'$ directly gives us a block match. 

For the running time, observe that we only use a constant factor more nodes and edges in $G_B$ than in our input graph $G'$. 
Further, we have $m = \setsize{V'}+4$ and $N \in \Theta(\setsize{V'}^2)$ because we unary encode the indices of the $\setsize{V'}$ nodes in the labels.
\end{proof}

Under the Triangle Detection hypothesis, we can infer that for dense block graphs \ie{$\setsize{E} \in \Theta(\setsize{V}^2)$} the $\cO(\max(\setsize{V}, m)^\omega)$-time algorithm of \Cref{thm:mm-easy} is conditionally optimal (up to subpolynomial factors). We get similar results for combinatorial algorithms using the combinatorial BMM conjecture.

\begin{corollary}
\label{cor:actual-LB-TD}
    There is no $\cO(\max(\setsize{V}, m)^{\omega - \varepsilon} + N)$-time algorithm for \bSMBG for $b \geq 4$ blocks and a constant $\smol > 0$ under the Triangle Detection hypothesis.
\end{corollary}

\begin{remark}
There is no combinatorial $\cO(\max(\setsize{V}, m)^{3 - \varepsilon} + N)$-time or $\cO((\setsize{E}m)^{1-\smol} + N)$-time algorithm for \bSMBG for $b \geq 4$ and a constant $\smol > 0$ under the combinatorial BMM conjecture.
\end{remark}
  
We also remark that one can infer similar lower bounds for sparse underlying block graphs based on the Strong Triangle Conjecture \cite{abboud2014PopularConjecturesImply}, stating that triangle detection needs time at least $\setsize{E}^{2\omega /(\omega + 1) - o(1) }$. 

This shows that \Cref{alg:fmm-overview} is optimal up to subpolynomial factors for $b \geq 4$ blocks. 
We note that our reduction uses \emph{duplicate} labels in the constructed block graph.
While a construction without duplicate labels that does not blow up the number of nodes remains unknown for $b=4, 5$ blocks, we give a simple alternative construction with $b=6$ blocks that avoids duplicate labels. 

\begin{lemma}
    \label{lem:td-lb-nodup}
    If we can solve \bSMBG for $b \geq 6$ blocks on a block graph with unique labels on a binary alphabet in time $T_6^{\text{uniq}}(\setsize{V}, N, m)$, we can solve TD in a graph with $n$ nodes in time $\cO(T_6^{\text{uniq}}(n, n^2, n) + n^2)$. 
\end{lemma}

\begin{proof}
The idea follows closely the proof of \Cref{thm:td-lb}. We focus on the changes in the construction.

Let $G'=(V',E')$ be the graph in which a triangle is to be detected. Assume that $G'$ does not have self-loops. 
We define a block graph $G_B = (V, E)$ with $b=6$ blocks $V = V_1 \sqcup \ldots \sqcup V_6$, each containing a copy of the nodes $V'$. 
We number the nodes $V'=\{v_1, \dots, v_n\}$ and write $v^j_i$ for the copy of $v_i$ in $V_j$.
We set the label of each node $v_i^j \in V^j$ to
\begin{align*}
    &\ell(v_i^1) \coloneqq 10^{i}, \\ 
    &\ell(v_i^2) \coloneqq 0^i, \\ 
    &\ell(v_i^3) \coloneqq 0^{2n-i+1}, \\ 
    &\ell(v_i^4) \coloneqq 0^{2n+i}, \\ 
    &\ell(v_i^5) \coloneqq 0^{4n-i+1}, \\ 
    &\ell(v_i^6) \coloneqq 0^{n-i}1.
\end{align*}
Clearly, there are no duplicate labels in the block graph.

The edges between blocks $V_1, V_2$ and $V_3, V_4$ as well as $V_5, V_6$ correspond to edges from $E'$.
For blocks $V_2, V_3$ and $V_4, V_5$, we add edges $(v^j_i, v^{j+1}_i)$ to $E$ for all $i \in [n], j \in \{2,4\}$ \ie{we only connect a node in $V_2$ ($V_4$) to its respective copy in $V_3$ ($V_5$)}.
We set the pattern to be $P \coloneqq 1\cdot 0^{9n+2} \cdot 1$.

For correctness, note that we can contract blocks $V_2, V_3$ and $V_4, V_5$, as there is only a single edge for each node, so that each resulting node within a block has the same label. The remaining proof works along the lines of the proof of \Cref{thm:td-lb}.
\end{proof}

\Cref{cor:actual-LB-TD} accordingly holds for algorithms on block graphs with distinct labels if we consider graphs with $b \geq 6$ blocks.
As we otherwise settle the complexity of the problem, this leaves the following narrow, yet enticing open question.

\begin{openquestion}
    Is there an $\cO(\max(\setsize{V}, m)^{\omega-\smol} + N)$-time algorithm, for any constant $\smol > 0$, that solves \bSMBG for $b \in \{4, 5\}$ exploiting unique labels? Or, is there a lower bound avoiding duplicate labels?
\end{openquestion}

\subsection{Logarithmic Number of Blocks}\label{sec:OV}
While our previous lower bound of \Cref{thm:td-lb} for graphs of $b \geq 4$ blocks rules out combinatorial improvements over the state-of-the-art $\cO(m\setsize E +N)$-time algorithm~\cite{DBLP:journals/jal/AmirLL00}, the SETH-based lower bound of~\cite{DBLP:journals/algorithmica/EquiNACTM23} also rules out non-combinatorial improvements.
However, their SETH-reducing instances require $b\in \Theta(|V|)$ many blocks. 
We give an improved lower bound, conditional on OVH and thus SETH, that uses only polylogarithmically many blocks.
Thus, we show that the problem becomes hard, to any type of algorithms, for much fewer blocks than previously known.

The formal definition and hardness assumption underlying OVH can be found in \Cref{sec:hardness}. 

\thmLogLB

\begin{proof}
    The proof is based on the ideas of \cite{DBLP:journals/siamcomp/BernardiniGPPR22} and \cite{DBLP:journals/algorithmica/EquiNACTM23}.
    In contrast to \cite{DBLP:journals/algorithmica/EquiNACTM23}, whose construction has unique node labels, we are able, for the more general setting, to construct an alignment gadget which uses $\cO(\log n)$ many blocks via a Fenwick-tree construction.\footnote{As Fenwick trees are well-known, we omit a full formal treatment. See \cite{fenwick1994} for the original paper and \cite{DBLP:journals/jfp/Yorgey25} for an exposition.}

    Consider an OV instance with two sets $A,B \subset \{0,1\}^d$ with $n$ elements each and ${d \in \omega(\log n)}$. 
    Without loss of generality, we may assume that $n$ is a power of $2$.
    We construct our pattern $P$ and block graph $G$ over the alphabet $\Sigma= \{0, 1, \$\}$.
    We first construct a gadget $G_A$, a block graph, for the set $A$ such that a single string $\beta \in B$ can be matched in the block graph $G_A$ if and only if there exists $\alpha \in A$ such that $\langle \alpha, \beta\rangle = 0$.
    Afterward, we add a prefix and a suffix alignment gadget to the block graph such that the resulting block graph includes any prefix and suffix of $P$ that consists of at most $n-1$ elements of $B$.
    
    We first define the pattern as \[
    P \coloneqq \$^{x} \cdot \beta_1 \cdot \$^x \cdot \beta_2 \cdot \$^x \cdot \ldots \cdot \beta_n \cdot \$^x,
    \]
    with $x = \log{n}$. 
    We now construct the set-gadget $G_A$. 
    First, we define a coordinate gadget $c$ with \[c(z) \coloneqq \begin{cases}
    \{0\}, & \text{if } z = 1, \\
    \{0, 1\}, & \text{if } z = 0, \\
    \end{cases}\] as shown on the left in \Cref{fig:G_A-gadgets}.
    We compose $d$ coordinate gadgets to construct an element gadget $g$, as shown in the middle of \Cref{fig:G_A-gadgets}, by setting $g(\alpha) \coloneqq c(\alpha_1)~c(\alpha_2) \ldots c(\alpha_d)$, where for each $i \in [d]$ we construct a block of nodes with labels from $c(\alpha_i)$, fully connected to its preceding block.
    For the set-gadget $G_A$, we set the element gadgets $\{g(\alpha) \mid \alpha \in A\}$ in parallel, as can be seen on the right in \Cref{fig:G_A-gadgets}.

    \begin{figure}
        \centering
        \includegraphics[width=0.9\linewidth,page=5]{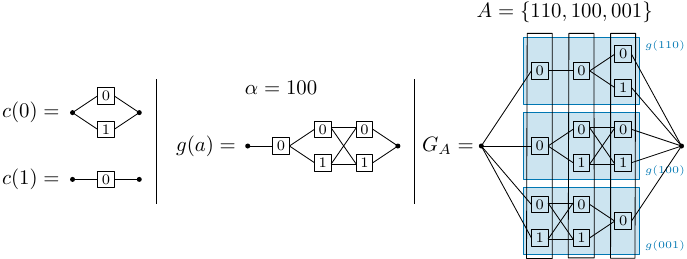}
        \caption{The coordinate, element, and set gadget on an example set $A$.}
        \label{fig:G_A-gadgets}
    \end{figure}
    
    \begin{figure}[ht]
        \centering
        \includegraphics[width=\linewidth,page=2]{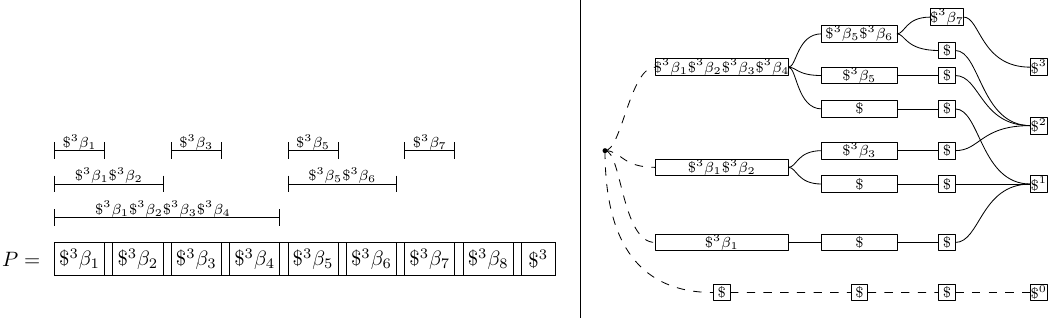}
        \caption{
        The construction of the prefix alignment gadget. Left, the subdivision of the pattern into ranges, forming a Fenwick tree, right, the resulting block graph. The dashed edges as well as the root and isolated nodes can be removed and are only shown for illustrative purposes. Note that any prefix $\$^3 \beta_1 \$^3 \ldots \beta_i \$^3$ of $P$ is included in the resulting block graph for $i \in \{0,\ldots, 7\}$. For $i=0$, such a path may start at the last shown block. As we must only allow proper prefixes of $P$ to be matched, we make it impossible to match $\beta_8$.
        }
        \label{fig:prefix-gadget}
    \end{figure}
    
    Next, we turn towards the prefix and suffix alignment gadgets. We describe here the prefix alignment gadget; the suffix gadget is constructed and works analogously.
    The gadget's idea is to include a path for any prefix \ie{for any of the first $\ell \in [n - 1]$ elements in $B$ in $P$}.
    This path will match the string $\$^x \beta_1 \$^x \ldots \beta_\ell \$^x$.
    We do so via a Fenwick tree (also known as Binary Indexed Tree); the basic idea is depicted in \Cref{fig:prefix-gadget}.

    To connect the prefix construction with the set-gadget $G_A$, we fully connect the last block of the prefix construction with the starting block of $G_A$.

    Analogously to the prefix alignment gadget, reversing the pattern and the resulting block graph, we construct the suffix alignment gadget and connect it to $G_A$.

    For correctness, assume that there exists $i, j \in [n]$ such that $\langle \alpha_i, \beta_j \rangle =0$. 
    We first take a path in the prefix gadget that covers the part $\$^x \beta_1 \ldots \beta_{j-1} \$^x$ in the pattern. By the Fenwick tree construction, such a path always exists.
    By construction of the gadget $G_A$, more specifically the element gadget for $\alpha_i$, the string $\beta_j$ is included in the block graph of $g(\alpha_i)$. 
    By the construction of the suffix gadget, there is a path that includes the string $\$^x \beta_{j+1} \ldots \beta_{n} \$^x$.
    As all three parts -- prefix, $G_A$, and suffix gadget -- are fully connected, there exists a path in $G$ that has a concatenated label including $P$, and the \bSMBG instance is a YES instance.

    Assume that there is a path in $G$ such that $P$ is included in its concatenated label.
    As both the prefix and the suffix alignment gadget are necessary for matching the full pattern, since both gadgets are too short, and $G_A$ does not include $\$^x$, any such path must span all three parts.  
    Because $G_A$ does not include the character $\$$ and the prefix and suffix gadgets start and end with a $\$$, the string included by $G_A$ must exactly match a $\beta \in B$. By construction of $G_A$, $\beta$ must be orthogonal to an element in $A$, thus there exists an OV witness.

    We analyze the size of our constructed instance. The pattern has length $nd+(n+1)x \in \Theta(n d)$.
    For the block graph, the prefix and suffix alignment gadgets both include $\log n$ many blocks of $\log n$ many nodes each. For their labels, their combined length does not exceed $|P| \log{n}$, as each part of the pattern is not included in more than $\log n$ many nodes by standard results on Fenwick trees.
    By the same standard results, each node has an out-degree of at most $\log n$, so we have $\Theta(\log^2n)$ edges in the prefix and suffix gadget.
    For the set-gadget $G_A$, each element gadget $g(\alpha)$ for $\alpha \in A$ can be considered independently. Each gadget, connected in parallel, introduces $d$ blocks with at most $2$ nodes and a single character per node. Thus in $G_A$ we have $d$ blocks with at most $2n$ nodes and combined $2dn$ label length.
    As the three parts are fully connected, in total $G$ has $|E| \in \Theta(nd + n\log n)$ many edges, as well as $b \in \Theta(d)$ many blocks, and $|V| \in \Theta(n d)$ many nodes with a combined label length of $N \in \Theta(n d\log n)$.

    Thus, assuming \bSMBG can be solved in time $\cO((m|E|)^{1-\varepsilon} + N)$ for a constant $\varepsilon > 0$ and the specific block number $b$ as constructed here, we can solve OV in time $\cO((n^2d^2 \log n)^{1- \varepsilon}) \leq \Oish(n^{2- \varepsilon})$ for $d \in \polylog(n)$, thus contradicting OVH. 
    Note that the complexity of \bSMBG is monotone in the number of blocks as we can always add additional blocks to our graph without asymptotically increasing the number of used nodes, edges, the pattern length, or the total length of the node labels. Thus, this gives us the same lower bound for \bSMBG for all $b \in \omega(\log n)$.
\end{proof}

\section*{Acknowledgments}
Paweł Gawrychowski was partially supported by the Polish National Science Centre grant number 2023/51/B/ST6/01505.

\bibliographystyle{alphaurl}
\bibliography{biblio}

\end{document}